%% file: __main.tex
\documentclass{article}
\usepackage{kr}

\usepackage{times}
\usepackage{soul}
\usepackage{url}
\usepackage[hidelinks]{hyperref}
\usepackage[utf8]{inputenc}
\usepackage[small]{caption}
\usepackage{graphicx}
\usepackage{amsmath}
\usepackage{amsthm}
\usepackage{booktabs}
\usepackage{algorithm}
\usepackage{algorithmic}
\input{_packages}

\input{_macros}

\title{About the Multi-Head Linear Restricted Chase Termination}

\author{%
Lukas Gerlach$^1$\and
Lucas Larroque$^2$\and
Jerzy Marcinkowski$^3$\and
Piotr Ostropolski-Nalewaja$^3$ \\
\affiliations
$^1$Knowledge-Based Systems Group, TU Dresden, Dresden, Germany\\
$^2$Inria, DI ENS, ENS, CNRS, PSL University, Paris, France\\
$^3$University of Wrocław, Poland\\
\emails
lukas.gerlach@tu-dresden.de,
lucas.larroque@inria.fr,
\{jma, postropolski\}@cs.uni.wroc.pl
}

\begin{document}

\maketitle

\input{01-introduction}

\input{02-preliminaries}
\input{03-mixed-chase_forest-version}

\input{03b-mixed-to-atomic-path}

\input{04-mso}

\input{06-magic-proof}
\input{07-conclusion}

\newpage

\section*{Acknowledgments}
Jerzy Marcinkowski and Piotr Ostropolski-Nalewaja were supported by
Polish National Science Center (NCN) grant  2022/45/B/ST6/00457

On TU Dresden side, this work was partly supported
by DFG (German Research Foundation) in project 389792660 (TRR 248, \href{https://www.perspicuous-computing.science/}{Center for Perspicuous Systems}) and in the CeTI Cluster of Excellence;
by the Bundesministerium für Bildung und Forschung (BMBF) 
under European ITEA project 01IS21084 (\href{https://www.innosale.eu/}{InnoSale, Innovating Sales and Planning of Complex Industrial Products Exploiting Artifical Intelligence},
in the \href{https://www.scads.de}{Center for Scalable Data Analytics and Artificial Intelligence} (ScaDS.AI), 
and in project 13GW0552B (\href{https://digitalhealth.tu-dresden.de/projects/kimeds/}{KIMEDS});
by BMBF and DAAD (German Academic Exchange Service) in project 57616814 (\href{https://secai.org/}{SECAI}, \href{https://secai.org/}{School of Embedded and Composite AI});
and by the \href{https://cfaed.tu-dresden.de}{Center for Advancing Electronics Dresden} (cfaed).

\appendix

\bibliographystyle{kr}
\bibliography{references}

\newpage
\input{A-second-main}

\newpage
\input{B-space-saver}
\input{05A-MSO-formulae}

\end{document}

%% file: _packages.tex
\usepackage{amssymb, mathtools}
\usepackage{etoolbox}
\usepackage[capitalize]{cleveref}
\usepackage{xcolor}
\usepackage{comment}
\usepackage{enumerate}
\usepackage{multicol}
\usepackage[normalem]{ulem}
\usepackage{xspace}
\usepackage{times}
\usepackage{thmtools}

\theoremstyle{plain}
\newtheorem{theorem}{Theorem}
\newtheorem{innercustomlem}{Lemma}
\newenvironment{customlem}[1]
{\renewcommand\theinnercustomlem{#1}\innercustomlem}
{\endinnercustomlem}
\crefname{innercustomlem}{Lemma}{Lemmas}

\crefname{innercustomthm}{Theorem}{Theorems}

\crefname{innercustomprop}{Proposition}{Propositions}

\newtheorem{innercustomobs}{Observation}
\newenvironment{customobs}[1]
{\renewcommand\theinnercustomobs{#1}\innercustomobs}
{\endinnercustomobs}
\crefname{innercustomobs}{Observation}{Observations}

\crefname{innercustomobs}{Corrolary}{Corrolaries}

\newtheorem{corollary}[theorem]{Corollary}
\newtheorem{lemma}[theorem]{Lemma}
\newtheorem{assumption}[theorem]{Assumption}

\newtheorem{conjecture}[theorem]{Conjecture}

\newtheorem{observation}[theorem]{Observation}
\crefname{observation}{Observation}{Observations}

\theoremstyle{definition}
\newtheorem{definition}[theorem]{Definition}
\newtheorem{example}[theorem]{Example}

\theoremstyle{remark}

\declaretheoremstyle[%
  spaceabove=-3pt,%
  spacebelow=4pt,%
  headfont=\normalfont\itshape,%
  postheadspace=1em,%
  qed=\qedsymbol%
]{betterstyle} 
\declaretheoremstyle[%
  spaceabove=-6pt,%
  spacebelow=4pt,%
  headfont=\normalfont\itshape,%
  postheadspace=1em,%
  qed=\qedsymbol%
]{noskipstyle}
\declaretheorem[name={Proof},style=proof,unnumbered,
]{prf}
\declaretheorem[name={Proof},style=noskipstyle,unnumbered,
]{pr}

%% file: _macros.tex
\renewcommand{\vec}[1]{\bar{#1}}

\newcommand{\thinexists}[1]{\exists{\scalebox{0.9}{$#1$}}\,}
\newcommand{\Exists}[1]{\exists{\scalebox{0.8}{$#1$}}.\,}

\newcommand{\Forall}[1]{\forall{\scalebox{0.8}{$#1$}}.\,}

\newcommand{\function}[1]{\mathit{#1}}

\newcommand{\predicate}[1]{\mathtt{#1}}

\newcommand{\apred}{\predicate{A}}

\newcommand{\arity}[1]{\function{ar}(#1)}

\newcommand{\structure}[1]{\mathbb{#1}}
\newcommand{\instance}{\structure{I}}
\newcommand{\inst}{\instance}

\newcommand{\adom}[1]{\function{adom}(#1)}

\newcommand{\singleton}[1]{\{#1\}}
\newcommand{\pair}[1]{(#1)}

\newcommand{\vk}{\smash{\vec{k}}}

\newcommand{\vx}{\vec{x}}

\newcommand{\vy}{\vec{y}}

\newcommand{\vz}{\vec{z}}

\newcommand{\vt}{\vec{t}}

\newcounter{rscounter}
\newcounter{dbcounter}
\newcommand{\rs}{\ruleset}

\newcommand{\iffi}{\textit{iff} }

\newcommand{\ruleapply}{\function{ruleApply}}
\newcommand{\apply}{\function{apply}}

\newcommand{\head}[1]{\function{head}(#1)}
\newcommand{\body}[1]{\function{body}(#1)}

\newcommand{\nats}{\mathbb{N}}

\renewcommand{\rs}{\ensuremath{\mathcal{T}}\xspace}

\definecolor{salmon}{HTML}{b35470}

\newcommand{\allterm}{\mathrm{CT}_{\forall}}
\newcommand{\allallterm}{\mathrm{CT}_{\forall\forall}}

\newcommand{\abstractatoms}{\mathfrak{A}}

\newcommand{\aamaxint}{k_{\max}}

\newcommand{\treesymbol}{\theforest}
\newcommand{\msotreesymbol}{\theforest}

\newcommand{\forestsymbol}{\mathbb{F}}
\newcommand{\addresses}{\mathbb{A}}

\newcommand{\theforestrun}[1]{\theforest_{#1}}

\newcommand{\deriv}{\mathcal{M}}
\newcommand{\derivp}{\mathcal{N}}
\newcommand{\derriv}{\mathcal{R}}
\newcommand{\demriv}{\mathcal{M}}

\newcommand{\depriv}{\mathcal{P}}

\newcommand{\derivstep}[1]{\mathbb{M}_{#1}}
\newcommand{\derivpstep}[1]{\mathbb{N}_{#1}}
\newcommand{\derrivstep}[1]{\mathbb{R}_{#1}}
\newcommand{\demrivstep}[1]{\mathbb{M}_{#1}}

\newcommand{\derivtrig}[1]{\pi^\deriv_{#1}}

\newcommand{\derivforest}{\derivstep{}}

\newcommand{\derrivforest}{\derrivstep{}}

\newcommand{\deprivforest}{\mathbb{P}}

\newcommand{\agraph}{\mathbb{G}}

\newcommand{\adr}[1]{\langle#1\rangle}
\newcommand{\absadr}[1]{[#1]}

\newcommand{\theabstree}{\msotreesymbol}
\newcommand{\thetree}{\treesymbol}
\newcommand{\thechase}{\adr{\theforest}}
\newcommand{\theforest}{\forestsymbol}

\newcommand{\presuperior}{\textcolor{gray}{\text{\raisebox{-0.1pt}{\resizebox{0.95em}{0.6em}{$\blacktriangledown$}}}}}

\newcommand{\absapply}{\function{abs\_apply}}

\newcommand{\flower}{P}
\newcommand{\defriv}{\mathcal{P}}

\newcommand{\infinite}{\mathsf{infinite}}
\newcommand{\isFlower}{\mathsf{infPath}}

\newcommand{\earliereq}{\sqsubseteq}
\newcommand{\earlier}{\sqsubset}

\newcommand*{\blocks}{\mathsf{blocks}}
\newcommand*{\pblocks}{\mathsf{\depriv\text{-}blocks}}

\newcommand{\childrenamount}{2}
\newcommand{\childrendots}{
	\ifnum \childrenamount=2
	,
	\else
	,\dots,
	\fi
}

\newcommand{\blocker}{\function{blockers}}

\newcommand{\adresy}{\inst{\mathbb A}^*}
\newcommand{\nN}{{n\in {\mathbb N}}}

\newcommand{\lesspretty}{\preccurlyeq}
\newcommand{\child}[1]{#1\text{-}\mathsf{child}}

\newcommand{\ppath}{\deprivforest^{\textnormal{\tiny 0}}}
\newcommand{\btr}{btr}
\newcommand{\ancestralPath}{\mathsf{ancPath}}

\newcommand{\preSet}{\mathsf{bfr}}
\newcommand{\freeSet}{\mathsf{any}}
\newcommand{\areTwins}{\mathsf{twins}}
\newcommand{\bt}[1]{[\![#1]\!])}
\renewcommand{\bt}[1]{(#1)}
\newcommand{\fragileSet}{\mathsf{fragile}}
\newcommand{\teamtriple}{\mathsf{team}}

\newcommand{\safe}{\mathsf{safe}}
\newcommand{\harma}{\mathsf{harm}_1}
\newcommand{\harmb}{\mathsf{harm}_2}

%% file: 01-introduction.tex
\begin{abstract}
The chase is a ubiquitous algorithm in database theory. However, for existential rules (aka tuple-generating dependencies), its termination is not guaranteed, and even undecidable in general. The problem of termination becomes particularly difficult for the restricted (or standard) chase, for which the order of rule application matters. Thus, decidability of restricted chase termination is still open for many well-behaved classes such as linear or guarded multi-headed rules. We make a step forward by showing that all-instances restricted chase termination is decidable in the linear multi-headed case.
\end{abstract}

\section{Introduction}

The \emph{chase} procedure was defined in late 1970s, and has been considered one of
the most fundamental database theory algorithms since then. It has been applied
to a wide spectrum of problems, including
query answering and
containment under constraints (\cite{AhSU79} and
\cite{CaGK13} among other examples)
computing solutions to data exchange problems \cite{FKMP05}
 and many others (see \cite{DeNR08} for details).
 And this is not just about theory: in several contexts
 the chase procedure has recently been
 efficiently implemented, see for example \cite{BKMM17}, \cite{Urban18}, {\cite{IGMSK2024}}.

 The reason for this ubiquity, as elegantly explained in \cite{DeNR08}, is that, for an
  input consisting of a database instance $\inst$
and  a set $\rs$ of existential rules (also known as tuple-generating dependencies, or TGDs),  \emph{chase} returns a finite instance  $Chase(\inst,\rs)$
being a universal model of $\inst$ and $\rs$, i.e., a model that can be homomorphically
embedded into every other model of $\inst$ and $\rs$. Intuitively,  $Chase(\inst,\rs)$  is a model
of $\inst$ and $\rs$ which  only contains information
already implicitly present in $\inst$ and $\rs$.

\noindent
But it only returns such a finite instance when it {\bf terminates}.

 \subsection{Chase in Five or Six Paragraphs}\label{sec:5-paragraps}

As we said, the input of  \emph{chase} comprises a database instance $\inst$ and a set $\rs$  of rules, which are formulas of the form:
 \begin{equation} \Forall{\vx \vy} \Phi(\vx, \vy) \to \Exists{\vz} \Psi(\vy, \vz) \label{eq:tgd} \end{equation}
where $\Phi$ and $\Psi$ are conjunctions
of atoms, known as the \emph{body} and the \emph{head} of this rule.

{At each execution step, the \emph{chase} extends the instance $\inst$ by adding new atoms. Specifically, \emph{whenever} there is a rule $\rho$ in $\rs$, as in (\ref{eq:tgd}), and tuples of terms $\bar{a}$ and $\bar{b}$ such that $\Phi(\bar{a}, \bar{b})$ holds in the current instance, a tuple $\bar{c}$ of fresh terms (called ``nulls'') is created, and new relational atoms are added to ensure that $\Psi(\bar{a}, \bar{c})$ also holds.\footnote{The tuple $\bar{z}$ can be empty if $\rho$ is a Datalog rule.} This process continues until a fixpoint is reached, yielding a model of both $\inst$ and $\rs$.}

There are many different ways to formalize this idea. In particular, there are {two} ways  the above \emph{whenever} can be interpreted, leading to two {fundamental variants} of the \emph{chase}.

 One variant is the
 \emph{restricted chase}, which is  lazy  - it only adds new atoms if  $\Phi(\bar a,\bar b)$  is true in the current structure, but $\exists \bar z \; \Psi(\bar a,\bar z)$ is {not}. \emph{Oblivious chase} is an eager version: it
always adds a new {witness $\Psi(\bar a, \bar c)$ - where $\bar c$ is a tuple of fresh terms} - for each tuple
$\bar a, \bar b$ such that  $\Phi(\bar a, \bar b)$ is true.

The order of execution does not matter for \emph{oblivious chase}. {No matter the order in which the witnesses are added,} we
eventually get the same structure.
But \emph{restricted chase} is non-deterministic -- different execution orderings in which tuples and rules are picked can lead
to different fixpoints.

Clearly,  \emph{restricted chase}, in general, builds
smaller instances than the oblivious one. It is very easy to
devise an example where, according to the rules of \emph{restricted chase}, no new atoms would be ever created,
 while  \emph{oblivious chase} builds an
infinite instance. Following \cite{GMP23}
consider  for example   $\inst = \{R(a,b)\}$ with   $\forall x,y (R(x,y) \to \exists z  R(x, z))$ as  the only
rule in $\rs$. \emph{Restricted chase} will detect
that $\inst$ already satisfies the rule, while \emph{oblivious chase} will not terminate and will build an
 instance $\{R(a,b), R(a, v_1), R(a, v_2), \ldots \}$.

\subsection{The Issue of Termination}

As we said above, for \emph{chase} to be useful, one needs to be sure that it terminates.
The termination normally considered in this context is  all-instance termination\footnote{Which may be seen as  a ``static analysis'' problem, since we analyze
$\rs$ before we know the data.}:
we want to know, for given $\rs$, if {\em chase} terminates on all instances $\inst$. This question again comes in two versions: all-instance termination of \emph{oblivious chase} or  all-instance termination of \emph{restricted chase}, which means that
\emph{restricted chase terminates} regardless of the $\inst$ and regardless of the execution order.

Surprisingly, only as late as in 2014 \cite{GM14} it was proven that{, in general, all-instance termination - both in the oblivious and restricted variants - is undecidable}. A stronger version of this result is presented in \cite{BFN}, where undecidability is shown even for relational schemas only containing binary relations.

On the other hand, really a lot of work was done to identify sufficient conditions on the rules 
which  guarantee  all-instance termination. The most well-known example of such condition is
\emph{weak-acyclicity}  \cite{FKMP05}. But there are numerous other examples too  ~\cite{DeTa03,FKMP05,DeNR08,GHKK+13,GrST11,Marn09,MeSL09,GC2023}.


The undecidability proofs given in~\cite{GM14,BFN,CGLT2025} construct  sets of rather complicated TGDs which go very far beyond the most popular
well-behaved classes of rules, like linear, guarded or sticky. This leads to a natural question, whether while undecidable in general, all-instance termination could be decidable
for sets of rules from such  classes\footnote{Notice that, unlike the ``sufficient conditions'' mentioned above, the well-behaved classes do not imply all-instance termination. They just, hopefuly, make all-instance termination decidable.}.

This question turned out to be relatively easy to answer for the case of \emph{oblivious chase}. For guarded TGDs, the problem is 2EXPTIME-complete, and becomes PSPACE-complete for linear TGDs~\cite{CaGP15}. The sticky case was addressed in~\cite{CaPi19}, where it is shown that the problem is PSPACE-complete.
Things are much more subtle in the case of \emph{restricted chase}, which is also more interesting in this context, because, as we explained,
\emph{restricted chase} is more likely to terminate than \emph{oblivious chase}.
Things are much more subtle because, as we explained before, in this case we need to deal
with all possible execution orderings, each of them leading to different result. But, despite of the difficulties, some progress was made. First it was shown that the problem is decidable for single-head linear\footnote{A rule like in (\ref{eq:tgd}) is linear if  $\Phi$ consists of a single atom and it is
single-head if $\Psi$ only comprises one atom.} TGDs~\cite{LMTU19}. Then decidability was proven, in \cite{GMP23}, for
guarded and sticky classes of TGDs, but also only in the single-head case.

 \subsection{Fairness. Why Is the Single-Head Case Easier?}
 One important nuance that we skipped in our brief presentation of \emph{chase} in  Section \ref{sec:5-paragraps} is fairness.
For \emph{restricted chase} to really create a model of $\inst$ and $\rs$, it must at some point apply each possible rule
  to every eligible $\Phi(\bar a,\bar b)$ (unless, of course, $\exists \bar z \; \Psi(\bar a,\bar z)$ gets satisfied earlier). Such execution orderings, when no
  rule application is starved, are called \emph{fair}. Let us illustrate it with two examples:

  \begin{example}
  Let $\rs$ consist of two rules:
  \begin{align*}
 \forall x, y \; R(x,y)&\;\;\to\;\; \exists z \;R(y,z)  \tag{$\rho_1$}\\
 \forall x, y \; R(x,y)&\;\;\to\;\;  R(y,y)  \tag{$\rho_2$}
  \end{align*}
 and let $\inst$ consist of a single atom $R(a_0,a_1)$.

 Then we can execute $\rho_1$ creating a new atom $R(a_1,a_2)$, then apply $\rho_1$ to this atom, creating
 $R(a_2,a_3)$ and so on, building an infinite instance. Such execution ordering  would not be fair: application
 of $\rho_2$ to  $R(a_0,a_1)$, and also to each of the created atoms, would starve.

 On the other hand, we could start from applying $\rho_2$ to $R(a_0,a_1)$, creating $R(a_1,a_1)$. In an instance comprising these
 two atoms both rules would be already satisfied, and \emph{restricted chase} would terminate at this point.

 But there is also another execution ordering here, which is fair and which leads to an infinite chase: each time, after you use $\rho_1$ to
 create $R(a_n,a_{n+1})$, apply $\rho_2$ to $R(a_{n-1},a_{n})$ and produce $R(a_{n},a_{n})$. \qed
\end{example}

 If the above example gives you the impression that analyzing fairness is complicated, it is the right impression. And,
 in  decidability proofs, like in \cite{LMTU19,GMP23} at some point you need to show that, if your termination criterion fails,
then there exists an instance leading to an infinite fair derivation.

For this reason, a  very important tool in both \cite{LMTU19} and \cite{GMP23} are Fairness Lemmas. In \cite{LMTU19} it is shown that, for $\rs$ being
a set of
linear single-head rules, if there exists an infinite \emph{restricted chase} then there exists also a fair one. In \cite{GMP23} the authors show
that it is actually true for {\bf any} set of {\bf single-head rules}. This allows them to just forget about the fairness condition.

But, as also noticed in both  cited papers,  this lemma does not hold for rules, even linear, which are multi-head:
{
 \begin{example}\label{ex:drugi}
  Let $\rs$ consist of two rules:
  \begin{align*}
      \forall x,z,u\; R(x,u,u) \;\;&\to\;\; \exists y\; R(x,y,u) \wedge R(y,u,u) \tag{$\gamma$}\\
      \forall y,z,u\; R(y,u,u) \;\;&\to\;\;  R(u,u,u) \tag{$\delta$}
  \end{align*}
 and let $\inst$ consist of a single atom $R(a,b,b)$. It is easy to see that there exists an infinite restricted derivation: rule $\gamma$ can be applied indefinitely to atoms of the form $R(t_n, b, b)$, where $t_0 = a$. However, such a derivation is not fair with respect to rule $\delta$, and it suffices to apply $\delta$ once to prevent $\gamma$ from ever being applied again.\qed
\end{example}}

For this reason, both the cited papers state decidability of restricted chase termination for multi-head rules as an open problem, whose solution would require new techniques. 

\subsection{Our Contribution}

Our contribution in this paper is
Theorem \ref{th:main}, saying that
all-instances restricted chase termination is decidable for rule sets consisting of linear multi-head rules.

To keep  the notations as light as possible,
our proof of Theorem \ref{th:main} is presented using:

\begin{assumption}\label{assumption:twoAtoms}
We will assume that all rule heads we consider contain exactly two atoms.
\end{assumption}

While this assumption is not free, as we do not have a procedure that normalizes rule sets into rules with exactly two atoms without impacting termination, the proof we present remains correct without this assumption, and no new arguments are needed.

\subsubsection*{Organization} In Sections \ref{sec:prelims}-\ref{sec:triggers} we introduce 
the basic notions and notations, which are in principle standard, but tailored to be useful in 
Sections \ref{sec:main-th}-\ref{sec:mso}. In Section \ref{sec:main-th} we formally state our main result, Theorem \ref{th:main}, and present an 
overview of the proof of this theorem. The proof can be in a natural way split into three parts, and the three parts are covered in Sections \ref{sec:mixed}-\ref{sec:mso}. Some proofs are deferred to the appendix.

%% file: 02-preliminaries.tex
\section{Preliminaries}\label{sec:prelims}


\smallskip
\noindent
Let $\mathfrak{P}$ be some finite set of \emph{predicates}, and let $\mathfrak{C}$, $\mathfrak{V}$, and $\mathfrak{N}$ be mutually disjoint, countably infinite sets of \emph{constants}, \emph{variables}, and \emph{labeled nulls} respectively.
The function $\function{ar}: \mathfrak{P} \to \mathbb{N}$ maps each predicate to its \emph{arity}. The \emph{arity} of $\mathfrak{P}$, denoted with $\arity{\mathfrak{P}}$, is the maximal arity of its predicates.
Elements from $\mathfrak{C} \cup \mathfrak{V} \cup \mathfrak{N}$ are called \emph{terms}. A term $t$ is \emph{ground} if $t \in \mathfrak{C} \cup \mathfrak{N}$.
An \emph{atom} is an expression of the form $P(\vt)$, where $P \in \mathfrak{P}$, $\vt$ is a list of terms and $\vert \vt \vert = \arity{P}$. Given an atom $\alpha$ we denote its  $i$-th term with $\alpha_{i}$.
A \emph{fact} is an atom with only ground terms. An \emph{instance} is a set of facts. A singleton instance is called \emph{atomic}.
The \emph{active domain} of a set (or a conjunction) $A$ of atoms is the set (denoted as $\adom{A}$) of terms appearing in atoms of $A$.

\smallskip
\noindent
\textbf{Existential Rules\;}
An FO formula of the form
$$\Forall{\vx \vy} \Phi(\vx, \vy) \to \Exists{\vz} \Psi(\vy, \vz)$$
is called an {\em existential rule},
where $\vx, \vy$ and $\vz$ are pairwise disjoint tuples of variables, and $\Phi$ and $\Psi$ are  conjunctions of atoms.
For a rule $\sigma$ as above, we denote $\Phi(\vx, \vy)$ as $\body{\sigma}$ and $\Psi(\vy, \vz)$ as $\head{\sigma}$.
In this paper we only  consider  {\em linear} rule sets, which means that the body of each rule contains at most one  atom.
The position in one of the atoms of a rule (both in the body or the head)  is called a {\em frontier position} if it
{contains} one of the variables from the tuple $\vy$.

Almost
\textbf{all the objects we define in Sections \ref{sec:prelims} - \ref{sec:triggers} depend on parameters: an instance $\inst$ and a rule set $\rs$.
We always assume that $\rs$ is a set of linear rules adhering to \cref{assumption:twoAtoms}.} This dependence is not reflected in the notation, as we do not believe this leads to misunderstandings.
%

\medskip
\noindent
\textbf{Homomorphisms.\;}
A homomorphism $h$ from a set of atoms $A$ to a set of atoms $B$ is a function from $\adom{A}$ to $\adom{B}$ such that for every atom $\alpha(\vx) \in A$ we have $\alpha(h(\vx)) \in B$. The notion of homomorphism naturally extends to queries, bodies, and heads of existential rules.

\smallskip
\noindent
\textbf{Adresses and Forests.\;} Now we introduce the language we will later use to address atoms in the (real) oblivious chase. Recall that
our rules are linear, so (intuitively) every atom $\alpha$ produced by \emph{chase} has a unique ``parent atom'', from whom $\alpha$ was created by some
rule of $\rs$, and the genealogy of such $\alpha$ can always be traced back to $\inst$.
%

\begin{definition}[Addresses]\label{def:address}
\begin{itemize}
\item Define the set ${\mathbb A}$ of {\em address symbols} as $\{ \rho_\iota: \rho\in \rs \text{ and } \iota\in \singleton{1,2}\}$. We always reserve the letter $\iota$ for an index ranging over the set $\singleton{1,2}$, and this is related to \cref{assumption:twoAtoms}.
\item  \emph{Address} is a word $\omega u\in \adresy$, i.e. a word whose first element is an atom from $\inst$, followed by a sequence of  address symbols. {We will usually use letters $u,w,v$, and sometimes also $s,r$, to denote words from $\inst{\mathbb A}^*$ or ${\mathbb A}^*$.}
\item For $w\in \adresy  $ and $a,b\in {\mathbb A}$ we say that $wa$ is an 
$a$-\emph{child} of $w$ (or just \emph{child of} $w$), and that $wa$ and $wb$ are \emph{siblings}.
\item By an \emph{infinite path} we mean an infinite sequence of adresses $w_0,w_1,w_2,\ldots$ such that $w_0\in \inst$ and $w_{n+1}$ is a child of $w_n$, for each $n$.

\end{itemize}
\end{definition}

\begin{definition}[Forests]\label{def:forest}
{A \emph{forest} is a set $\agraph$ of addresses such that}
{$\inst \subseteq \agraph \subseteq \inst{\mathbb A}^*$, and which is \emph{prefix-closed}, meaning that if $\omega uv\in \agraph$ for some $\omega\in \inst$ and $uv\in {\mathbb A}^*$ then also $\omega u\in \agraph$}.

If $\inst$ is a singleton we use the term \emph{tree} instead of \emph{forest}.
For a forest $\agraph$ and $u\in \agraph$ we use the notation $\agraph_u$ for the set $\singleton{uv \mid uv\in \agraph}$ of all the descendants of $u$ in $\agraph$. 
\end{definition}

\section{The Real Oblivious Chase}\label{sec:real-chase}

In this section, following \cite{GMP23}, we define the (real) \emph{oblivious chase}, a static object, that serves as an arena where restricted chase takes place.

\begin{definition}[Rule application]\label{def:rule-app}
Let $\alpha$ be an atom, let $\rho\in \rs$ be a rule, such that there exists a homomorphism $h$ from $\body{\rho}$ to $\alpha$, and let $u$ be an address. Then $\ruleapply(\rho, \alpha, u)$ is a pair of atoms obtained from $\head{\rho}$ by replacing each universally quantified variable $x$ of $\rho$ with $h(x)$ and each existentially quantified variable $z$ with a
labeled null of the form $z_u^\rho$. We denote the $\iota$-th atom of $\ruleapply(\rho, \alpha, u)$ with $\ruleapply_\iota(\rho, \alpha, u)$.
\end{definition}

The above definition allows us to precisely identify when each labeled null was created. It will be useful in \cref{sec:mixed}.

The following definition is one of the most important:

\begin{definition}
The \emph{representation of the oblivious chase} (or simply \emph{representation}) is a pair consisting of a forest $\theforest$ and a labeling function $\adr{\_}$  from  $\theforest$ to ground atoms. We define $\theforest$ and $\adr{\_}$ by a simultaneous induction:

\begin{itemize}
  \item if   $\omega \in \inst$ then   $\omega\in \theforest$, and $\adr{\omega} = \omega$;
  \item If $u\in\theforest$, $\rho$ is a rule of $\rs$, and there is homomorphism $h$ from $\body{\rho}$ to $\adr{u}$ then for
  each  $\iota$ it holds that $u\rho_\iota\in \theforest$  and $\adr{u\rho_\iota} = \ruleapply_\iota(\rho, \adr{u}, u)$.
\end{itemize}
\end{definition}

Note that this definition deliberately avoids discussing the notion of triggers and their application order. This is for a good reason: as we said before, the oblivious chase is defined statically, and serves as a scene for the restricted chase.


We will often need to talk about the set of atoms represented by a forest. Therefore, by slightly abusing notation:

\begin{definition}[Materialization]
Given a forest ${\mathbb{G}}\subseteq \theforest$ we define $\adr{\mathbb{G}}$ as $\{\adr{u}: u\in \mathbb{G}\}$.
\end{definition}

This allows us to cleanly define the oblivious chase, and later, the restricted chases.

\begin{definition}
The \emph{oblivious chase} is defined as $\adr{\theforest}$.
\end{definition}

Chase, and restricted chase in particular, is all about atoms. But each atom of interest is somewhere in $\adr{\theforest}$ and {\bf the only way we will ever refer to atoms will be by their addresses} in $\theforest$. 

Notice  that  the labeling function  $\adr{\_}$ does not need to be one-to-one. {To visualize this, consider Example 2 and addresses of the form $(R(a, b, b))(\gamma_2)^*\delta_1$. Each such address is labeled with $R(b,b,b)$}. For this reason  \cite{GMP23} think of  $\adr{\theforest}$ as a multiset and call it \emph{the real oblivious chase}.

\section{Triggers and Chase Derivations}\label{sec:triggers}

A trigger is an address where an atom lives, and a rule which can be applied to this atom:

\begin{definition}[Triggers]
\begin{itemize}
\item A trigger is a pair $\pair{\rho, u}$ consisting of a rule $\rho \in \rs$ and of a $u\in \theforest$ such that
there exists a (unique) homomorphism from $\body{\rho}$ to $\adr{u}$.

\item
For a trigger $\pi$, we denote as $\apply(\pi)$ the set of addresses $\singleton{u\rho_1, u\rho_2}$ and with $\apply_\iota(\pi)$ the address $u\rho_\iota$.

\item
For a forest $\agraph\!\subseteq\!\theforest$, and a trigger $\pi \!=\! \pair{\rho, u}$, 
we say that $\pi$:
\begin{itemize}
  \item \emph{appears} (or \emph{is}) \emph{in} $\agraph$
    \;\iffi $u \in \agraph$. We also write $\pi \in \agraph$.
  \item is \emph{active} in $\agraph$ \;\iffi $\pi \in \agraph$ and ${\apply(\pi)} \not\subseteq \agraph$.
  \end{itemize}
\end{itemize}
\end{definition}

{Note that every trigger (by definition) appears in $\theforest$. However, not every pair $\pair{\rho, u}$ qualifies as a trigger. For such a pair to be a trigger, there must exist a homomorphism from the body of $\rho$ to $\adr{u}$. Since the bodies of rules consist of single atoms, at most one such homomorphism can exist.}


The very idea of restricted chase is that blocked triggers are never applied:

\begin{definition}[{Trigger blocking}]\label{def:blocking}{
  The pair $\bt{v_1,v_2}\in\agraph^2$ is a \emph{blocking team} for a trigger $\pair{\rho, u}$ in $\agraph\subseteq\theforest$ \iffi there is a homomorphism $h$ such that $h(\adr{u})=\adr{u}$ and $h(\adr{u\rho_\iota})=\adr{v_\iota}$ for all $\iota$. A trigger that has a blocking team is \emph{blocked}.}
\end{definition}

In other words, $\bt{v_1,v_2}$ is a blocking team for $\pi$ if atoms $\adr{v_1},\adr{v_2}$ can witness that the existential rule
$\rho$ is  satisfied, when applied to the atom  $\adr{u}$. The following may be trivial to see, but is important to realize:

\begin{observation}\label{obs:trivial}
Let $\rho$ be a rule of the form
$$\Forall{\vx \vy} \phi(\vx, \vy) \to \Exists{\vz} \psi_1(\vy, \vz)\wedge \psi_2(\vy, \vz)$$ for some
atoms $\phi, \psi_1, \psi_2$.
Then a pair $\bt{v_1,v_2}$ is a blocking team for trigger $\pi = \pair{\rho, u}$ if and only if, for each $\iota,\iota'$:
\begin{enumerate}[(i)]
    \item the predicate symbols of  $\adr{v_\iota}$ and of $\adr{u\rho_\iota}$ are equal;
    \item if $j$ is a frontier position in $\psi_\iota$ then $\adr{u\rho_\iota}_j= \adr{v_\iota}_j$;
    \item if $j$ is any  position in $\psi_\iota$ and $j'$ is any  position in $\psi_{\iota'}$ and $\adr{u\rho_\iota}_j=\adr{u\rho_{\iota'}}_{j'}$ then also $\adr{v_\iota}_j=\adr{v_{\iota'}}_{j'}$.
\end{enumerate}
\end{observation}

Intuitively, the observation says that atoms in the blocking team of $(\rho,u)$ must contain at least the same positive
 information as the atoms $\adr{u\rho_\iota}$ do. But they may
 contain more: they may have more equalities between terms, and they can have something
 non-anonymous (for example constants from $\inst$) on positions where $\adr{u\rho_\iota}$ have new anonymous terms. This makes them
 potentially more ``capable'' of producing new atoms, and also of blocking, than  $\adr{u\rho_\iota}$ are.

\newcommand{\III}{\mathbb I}

Now, we introduce the notion of derivation, which intuitively corresponds to a growing sequence of forests that results from a sequence of trigger applications.

\begin{definition}[Chase derivations]\label{def:chase-derivations}
A {\emph{chase derivation $\deriv$ (of the rule set $\rs$ over the instance $\inst$)}} is a (finite\footnote{Since we mainly consider infinite derivations, our notation is tailored for such derivations. In order to talk about finite derivations, $\mathbb N$ would need to be replaced with its initial segment.  } or infinite) sequence of triggers $\singleton{\derivtrig{n}}_\nN$ such that there exists a sequence of forests $\singleton{\derivstep{n}}_\nN$ called \emph{derivation steps} where:
\begin{itemize}
  \item $\derivstep{0} = \inst$;
  \item trigger $\derivtrig{n}$ is {active} in $\derivstep{n}$;
  \item $\derivstep{n + 1} = \apply(\derivtrig{n}) \cup \derivstep{n}$.
\end{itemize}
For a derivation $\deriv$, we denote the forest $\bigcup_{n\in\nats} \derivstep{n}$ by $\derivforest$. 

\end{definition}

Clearly, for each derivation  $\deriv$ there is  $\derivforest\subseteq \theforest$. 

\begin{definition}[Derivation Types]
Derivation $\deriv$ is:
\begin{itemize}
  \item \emph{oblivious} when $\derivforest=\theforest$;
  \item \emph{restricted} when no trigger $\derivtrig{n}$ is blocked in $\derivstep{n}$;
  \item \emph{fair} when each active trigger in $\derivforest$ is blocked in $\derivforest$.
\end{itemize}
\end{definition}
\noindent

Informally speaking, a fair restricted derivation is a sequence in which all active triggers eventually either get applied or become blocked. Notice also (still being informal) that all oblivious derivations are permutations of the same set of all triggers, and that such a permutation is an oblivious derivation whenever it satisfies the
obvious constraint that trigger $(\rho,u)$ may not precede the creation of $u$.

\section{Stating the Main Theorem}\label{sec:main-th}

$\allterm(\inst)$ is the class of rule sets $\rs$ such that all fair restricted derivations of $\rs$ over instance $\inst$ are finite.
$\allallterm$ is the class of rule sets that belong to $\allterm(\inst)$ for all finite instances $\inst$.

\begin{theorem}[Main Theorem]\label{th:main}
  Membership in $\allallterm$ for sets of multi-headed linear existential rules is decidable.
\end{theorem}

\noindent
We dedicate the rest of this paper to the proof of Theorem \ref{th:main}. Consider a rule set $\rs$, which will now remain fixed.

Let us now begin our proof with a high-level overview.

\subsection{Overview of the Proof of Theorem \ref{th:main} }

Our proof of Theorem \ref{th:main} relies on three main lemmas.

First, in Section \ref{sec:mixed} we define (for a fixed $\inst$) an object called mixed derivation and we show our First Main Lemma:
 for each  $\inst$, existence of a mixed derivation is equivalent to the existence of an infinite, fair, restricted derivation. 

Then, in Section \ref{sec:atomic} we define, for an atom $\omega$, another object, called $\omega$-path-sensitive-derivation. This definition will not depend on $\inst$. As our Second Main Lemma, in this section, we prove that existence of a mixed derivation for \emph{any} $\inst$ is equivalent to existence of an
$\omega$-path-sensitive derivation.

Finally, in Section \ref{sec:mso} we show how to write a Monadic Second Order Logic formula $\Omega_\omega $ (which of course depends on $\rs$), which is true on the tree $\omega{\mathbb A}^*$ (with the relations $\child{\rho_\iota}$ for $\rho_\iota\in \mathbb A$) if and only if there exists an $\omega$-path-sensitive derivation (this equivalence is stated as our Third Main Lemma).
Since there are only finitely many pairwise non-isomorphic atoms $\alpha$, this, together with the well-known fact that Monadic Second Order Logic is decidable over such infinite trees, gives  the decision procedure, and  completes the proof of Theorem \ref{th:main}. Notice that
our entire decision procedure is confined to Section \ref{sec:mso}. The only purpose of Sections \ref{sec:mixed} and \ref{sec:atomic} is to prove correctness of this procedure.

%% file: 03-mixed-chase_forest-version.tex
\section{Mixed Derivations}\label{sec:mixed}

%
%

A finite instance $\inst$ remains fixed throughout this section.

\begin{definition}\label{def:mixed-chase}
A \emph{mixed derivation} is a pair $\pair{\demriv, \derriv}$ where $\demriv = \{\pi^\demriv_n\}_{n\in \mathbb N}$  is an infinite oblivious derivation, and
$\derriv$ is an infinite fair and restricted derivation, such that:
\begin{itemize}
\item $ \derriv$ is a subsequence of $\demriv$;
\item if $\pi^\demriv_n\in \derriv$ for some $n$ then $\pi^\demriv_n$ is not blocked in  $\demrivstep{{n}}$.
\end{itemize}
\end{definition}


\begin{lemma}[First Main Lemma]\label{trm:restrictedChaseToMixedChase}
  The following statements are equivalent:
  \begin{enumerate}
    \item There exists an infinite restricted and fair derivation.
    \item There exists {a mixed derivation}.
  \end{enumerate}
\end{lemma}

We devote the rest of this section to the proof of the above theorem. The  ($2 \Rightarrow 1$) direction
is trivial, so it is just the
$(1 \Rightarrow 2)$ implication that we need to show. The proof, we believe, is not completely straightforward. {Thus,} we kindly ask the Reader to brace for complications,  and nuances.

Suppose  $\derriv$ is a fair infinite restricted derivation and let $\derrivforest$ be its associated forest of addresses
(as per \cref{def:chase-derivations}).

In what follows, we aim to construct an oblivious derivation $\demriv$ such that $\pair{\demriv, \derriv}$ is indeed a mixed derivation.

What does it actually mean to construct this $\demriv$?

$\demriv$, as we know, is a permutation of all triggers. For $\derriv $ to be a subsequence of $\demriv$ we need to somehow start from
$\derriv $ and then slot the triggers from $\demriv \setminus \derriv$ between\footnote{By $\demriv \setminus \derriv$ (which does not type, as
 $\demriv$ and $\derriv$ are sequences, not sets)
we mean the set of triggers that do not appear in  $\derriv$.} the triggers in the sequence $\derriv$. Clearly, we need to 
ensure that all triggers $(\rho,u)$ are slotted after $u$ is created, as stated at the end of \cref{sec:triggers}.
What else can possibly go wrong here?

Think about the last item in Definition \ref{def:mixed-chase}, and imagine we have a trigger $\pi$ somewhere in the sequence
$\derriv$. We may need to slot some triggers from $\demriv \setminus \derriv$ ahead of $\pi$. 
However, slotting triggers introduces new blocking teams from $\theforest$. We thus need to be very careful to
make sure that no blocking team for $\pi$ is 
{created by} the triggers slotted ahead of $\pi$. 
One may think
 --- and, thanks to the assumption that $\derriv$ is restricted, this is actually true, {albeit not obvious} --- that each blocking team can only block finitely many triggers from $\derriv$. Thus, maybe we could simply, for
each such blocking team, 
add its members  after all its blockees from $\derriv$? It is not so simple, because while each blocking team {does} only block finitely many triggers, there may potentially be addresses participating in infinitely many blocking teams. To visualize this, consider again Example 2, and addresses of the form $(R(a,b,b))(\gamma_2)^*\delta_1$. Each such address participates in infinitely many blocking teams of triggers of the form $\pair{\gamma,\; (R(a,b,b))(\gamma_2)^*}$.

\subsection{ The Border of \texorpdfstring{$\derrivforest$}{R} and Beyond}
Since $\derriv$ is restricted, it might happen for some address $u$ of $\theforest$ that while $u \in \derrivforest$ there is a trigger $\pi = \pair{\rho, u}$ such that $u\rho_1,u\rho_2 \not\in \derrivforest$ as $\pi$ is blocked by some blocking team of addresses of $\derrivforest$. {Intuitively, we can imagine, that the} border between $\derrivforest$ and $\theforest \setminus \derrivforest $ runs  between such $u$ and its $\rho$-children



\begin{definition} A trigger $\pair{\rho, u}\not\in \derriv $ is called
\emph{border} trigger if $u\in \derrivforest$, and is called
 \emph{problematic} trigger if $u\not\in \derrivforest$.
\end{definition}

While at this point it may not be totally obvious for the Reader how problematic triggers earned their name, we will clarify this after \cref{def:better}.
\smallskip

Next we introduce some terminology for addresses in $\theforest\setminus \derrivforest$. Notice that for each such address $v$ there is  a (unique)
border trigger $\pair{\rho, u}$ such that $v=u\rho_\iota w$ for some $\iota$ and $w$.

\begin{definition}[Virtual addresses and ranks]\label{def:virtual}
  Addresses in $\theforest\setminus \derrivforest$ are called \emph{virtual}. When we say that $u\rho_\iota w$ is a virtual address, we think that $\pair{\rho, u}$ is a border trigger; in such cases, we sometimes say $u\rho_\iota w$ \emph{belongs to} $u$.

  We define the \emph{rank} of an address, by $rk(u)=0$ if $u\in\derrivforest$ and $rk(u\rho_\iota w)=|\rho_\iota w|$ if $u\rho_\iota w$ is virtual.
\end{definition}

So, function $rk$ (rank) says how far an address is from $\derrivforest$.

\begin{definition}
  Any term that appears in some atom of $\adr{\derrivforest}$ is called \emph{global}. Any other term from $\adr{\theforest}$ is \emph{local}.
\end{definition}

We end this section with an obvious observation, which will be very useful soon. 
Claim (iii) says that each local term is localized to the pair of subtrees of $\theforest$ rooted in $\apply(\pi)$
for exactly one border trigger $\pi$. Notice, however, that global terms can appear in multiple such subtrees.

\begin{observation}\label{lem:local-terms-belong-to-a-border-trigger}
\begin{enumerate}[(i)]
\item Suppose $s,sw\in\theforest$ and a term  $t$ occurs in both $\adr{s}$ and $\adr{sw}$.
Then for each prefix $v$ of $w$ term $t$ occurs in $\adr{sv}$.
\item Suppose $s,w\in\theforest$ are s.t. {neither} of them is a prefix of the other, and some term $t$ occurs in both $\adr{s}$ and $\adr{w}$.
Then there are prefixes $w'$ and $s'$ of $w$ and $s$ that are siblings, and such that $t$ occurs in both $\adr{w'}$ and $\adr{s'}$.
\item For each local term $t$ there exists exactly one border trigger $(\delta,v)$ such that if $t$ appears in some atom $\adr{w}$ then
$w\in \theforest_{v\delta_1}\cup  \theforest_{v\delta_2}$.
\end{enumerate}
\end{observation}

\subsection{Finding Better Versions of Elements of \texorpdfstring{$\theforest$}{F}}

We assumed that  $\derriv$ is a restricted {and} {\bf fair} derivation. So the only reason for a trigger $\pi$ to be a border trigger
is that it was blocked by some blocking team $\bt{v,v'}$ of $\derrivforest$. Note that $\bt{v,v'}$ may not be the only blocking team for $\pi$ in $\derrivforest$, but for our reasoning we need to consider any such team:

\begin{definition}\label{def:blocker}
We define $\blocker$ as  a function assigning to each border trigger $\pi$ a blocking team $\bt{v_1,v_2}$ for  $\pi$, such that   $v_1,v_2\in \derrivforest$.
\end{definition}

Recall the remark we made after Definition \ref{def:blocking}, that members of blocking teams are ``more capable'' than the blockees.
We will now formalize this intuition. We begin with:

\begin{definition}\label{def:better}
Suppose $\pi=(\rho,u)$ is a border trigger and  $\blocker(\pi)=\bt{v_1,v_2}$. Then we define $\btr(u\rho_\iota)=v_\iota$.
\end{definition}

In the above definition $\btr$ is short of \emph{better}.
Now we might be ready to explain why the non-border triggers of $\demriv \setminus \derriv$ are more problematic than border triggers.
In both cases we need to slot them between the triggers of $\derriv$. However,
if we slot a border trigger $(\rho,u)$ after
{addresses} $\btr(u\rho_1)$ and $\btr(u\rho_2)$ were already created in $\derriv$, then  $u\rho_1$ and $u\rho_2$ will not
block any trigger which comes later in  $\derriv$. \linebreak
Why? Because $\adr{\btr(u\rho_1)}$ and $\adr{\btr(u\rho_2)}$ are better at blocking than
$\adr{u\rho_1}$ and $\adr{u\rho_2}$, and they were already in $\adr{\derrivforest}$ at this point without blocking anyone, since $\derriv$ is restricted.

\begin{figure}
    \centering
    \includegraphics[width=0.75\linewidth]{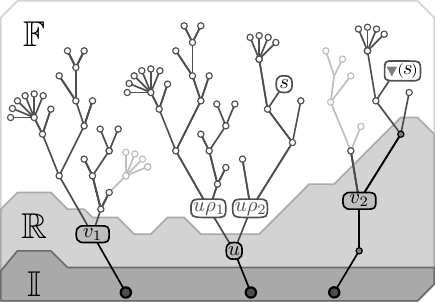}
    \caption{Section \ref{sec:mixed} in one picture. You can see a border trigger 
$\pi=(\rho, u)$: address $u$ is still in $\derrivforest$ while both $u\rho_\iota$ are right behind the border. The pair $\pair{v_1,v_2}$ of addresses in $\derrivforest$ is a blocking team for $\pi$, and $v_\iota=\btr(u\rho_\iota)$.
One can also see that the tree rooted in $v_\iota$ contains (possibly among other addresses) a copy of the tree rooted 
in $u\rho_\iota$ (as Lemma \ref{lem:prettier2}(B) postulates).
}
    
    \label{fig:enter-label}
\end{figure}

On the other hand, problematic triggers do not necessarily have blocking teams in $\derrivforest$.
Thus, for a problematic trigger $\pi=(\rho,u)$,  it is not obvious how to find
\emph{better} versions of $u\rho_\iota$ in $\derrivforest$. In consequence,
there is no natural moment in the
 derivation $\derriv$ when we could already be sure that {addresses} produced by $\pi$ can no longer block any future
 triggers from $\derriv$, so
 $\pi$ can  safely be slotted at that moment. 

Now, we are going to extend the function $better$ to addresses produced by problematic triggers.
As we said in \cref{sec:triggers}, members of a blocking team are also more capable of producing new {addresses} than their blockees.
 This will be formalized in \cref{lem:prettier2}.

{Recall that given an atom $\adr{u}$, with $\adr{u}_i$ we denote the term at position $i$ in $\adr{u}$.}
\begin{definition}\label{def:prettier}
For $u,w\in\theforest$ we say that $w$ is \emph{prettier} than $u$, denoted as $ u \lesspretty {w}$, if there exists
a homomorphism from $\adr{u}$ to  $\adr{w}$ which is an identity on global terms.
\end{definition}

In other words, $ u \lesspretty {w}$ if and only if:
\begin{enumerate}[(i)]
\item $\adr{u}$ and  $\adr{w}$ are atoms of the same relation;
\item if $\adr{u}_i$ is a global term then $\adr{u}_i = \adr{w}_i$;
\item if $\adr{u}_i=\adr{u}_j$  then $\adr{w}_i=\adr{w}_j$.
\end{enumerate}

Now recall Observation \ref{obs:trivial}. We are going to use it a lot.

\begin{lemma}\label{lem:prettier2}
Let $(\rho,u)$ be a border trigger. If $u\rho_\iota w\in\theforest$, then $(\btr(u\rho_\iota)) w \in\theforest$ and $ u\rho_\iota w\lesspretty (\btr(u\rho_\iota)) w$.
\end{lemma}

\begin{proof}

Let us start from a proof that $ u\rho_\iota \lesspretty \btr(u\rho_\iota)$. From \cref{def:better}, there is a homomorphism from $u\rho_\iota$ to $\btr(u\rho_\iota)$, so we just need to check whether it is an identity on global terms. First, notice that (since  $(\rho,u)$ is a border trigger), if $\adr{ u\rho_\iota}_i$ is a global term, then position $i$ is a frontier position in the body-atom of $\rho$. By \cref{obs:trivial} (ii) and definition of $\btr$, we get $\adr{ u\rho_\iota}_i = \adr{\btr( u\rho_\iota)}_i$ as required.

Since prettier parents produce prettier children, the lemma then follows by easy induction on the length of $w$.
\end{proof}

It follows from Lemma \ref{lem:prettier2} that $\presuperior$ is a function on $\theforest$:
\begin{definition}\label{lem:def-presuperior}
\phantom{a}\vspace{-6mm}
\begin{align*}
\phantom{a}\hspace{1.4cm} \presuperior(u) &= u, \tag{For $u\in \derrivforest$}\\
\phantom{a}\hspace{1.4cm} \presuperior(u\rho_\iota v) &= (\btr(u\rho_\iota))v \tag{For virtual $u\rho_\iota v$}
\end{align*}
\end{definition}




The Reader may now want to find the virtual address $s$ in Figure \ref{fig:enter-label},
and its image $\presuperior(s)$. Clearly, $rk(\presuperior(s))<rk(s)$, and it is no coincidence.
Indeed, function $\presuperior$, when applied to virtual addresses, is strictly monotonic with respect to rank: it moves addresses closer and closer to $\derrivforest$. In consequence, whatever address we start from, after a
finite number of applications,  $\presuperior$ will map this address to $\derrivforest$:

\begin{lemma}\label{lem:ranks-down}
  If $w\in\theforest$ is a virtual address, then $rk(\presuperior({w}))<rk({w})$.
  For each $w\in \theforest$ there exists a natural number $m_w$ such that $\presuperior^{m_w}(w)\in \derrivforest$.
\end{lemma}
\begin{proof} For the proof of the first claim {let $w = u\rho_\iota v$ where $\pair{\rho, u}$ is a border trigger}, and recall that,
by Definition \ref{def:virtual}, $ rk(u\rho_\iota v)= |\rho_\iota v| $. On the other hand, $\btr(u\rho_\iota)$  is an address in  $\derrivforest$, so
$rk((\btr(u\rho_\iota))v)\leq |v|$. Second claim follows immediately from the first.
\end{proof}

We now can extend our definition of $\btr$ to the entire $\theforest$. Your better {address} in $\derrivforest$  is the one
 where function $\presuperior$ ultimately maps you to:

\begin{definition}
For a problematic trigger $(\rho,u)$, define $\btr(u\rho_\iota)=\presuperior^{m_{u\rho_\iota}}(u\rho_\iota)$, and for any $v \in \derrivforest$ let $\btr(v) = v$.
\end{definition}

We are ready for our crucial Lemma and its Corollary.

\begin{lemma}\label{lem:crucial}
Suppose $\pi$ is a trigger from $\derriv$ and  $\bt{w_1,w_2}$ is a blocking team of $\pi$ in $\theforest$. Then  $ \bt{\presuperior(w_1),  \presuperior(w_2)}$ is also a blocking team for $\pi$.
\end{lemma}

\begin{corollary}\label{lem:so-useful}
Suppose  $\pi$ is a trigger from $\derriv$ and $\bt{w_1,w_2}$ is a blocking team of $\pi$ in $\theforest$. Then  $ \bt{\btr(w_1),  \btr(w_2)}$ is also a blocking team for $\pi$.
\end{corollary}

Now, in order to end the proof of our First Main Lemma, we need to do two things:
show how Corollary \ref{lem:so-useful} implies First Main Lemma (which is now easy, see Section \ref{sec:creating-mixed}) and
prove Lemma \ref{lem:crucial} (which we do not think
looks  obvious at this point)
which is done in the  next two subsections.


\subsection{Proof of  Lemma \ref{lem:crucial}. The Easy Part}

Throughout this subsection, we consider a border trigger $(\delta,v)$.   
Recall, from Lemma \ref{lem:prettier2}, that for each $w$, if the address $v\delta_\iota w$ is in $\theforest$ then the address $\btr(v\delta_\iota) w$ also is {in $\theforest$}, and is prettier than $v\delta_\iota w$. So, one may think of 
$\theforest_{\btr(v\delta_\iota)}$ as of a prettier copy of 
the set $\theforest_{v\delta_\iota}$.

We will begin this section with three very easy observations.
The general message  is that if two atoms in 
$\adr{\theforest_{v\delta_\iota}}$ share terms then the respective atoms in 
$\adr{\theforest_{\btr(v\delta_\iota)}}$ also share terms at the respective positions.

The observations will be followed by one lemma, which is not difficult either, but there is certain depth there.

First, notice that that the prettier you are the more terms you share with your children:

\begin{observation}\label{obs:easy1} Let $w\in {\mathbb A}^*$ and $a\in {\mathbb A}$ and suppose that
$\adr{v\delta_\iota w}_i= \adr{v\delta_\iota wa}_{i'}$ 
Then also $\adr{{\btr(v\delta_\iota)} w}_i= \adr{{\btr(v\delta_\iota)} wa}_{i'}$.
\end{observation}
\begin{proof}
    Follows from \cref{lem:prettier2} and \cref{def:rule-app}.
\end{proof}

And the prettier your parent is the more terms you share with your sibling:

\begin{observation}\label{obs:easy2} Let $w\in {\mathbb A}^*$ and $a,a'\in {\mathbb A}$. If   $\adr{v\delta_\iota wa}_i = \adr{v\delta_\iota wa'}_{i'}$ \;then also\; $\adr{\btr(v\delta_\iota) wa}_i = \adr{\btr(v\delta_\iota) wa'}_{i'}$.
\end{observation}
\begin{proof}
    Follows from \cref{lem:prettier2} and \cref{def:rule-app}.
\end{proof}
By inductive application of \cref{obs:easy1} we get:

\begin{observation}\label{obs:easy3}
Let $w,w'\in {\mathbb A}^*$. If $\adr{v\delta_\iota w}_i= \adr{v\delta_\iota ww'}_{i'}$ then
$\adr{\btr(v\delta_\iota) w}_i= \adr{\btr(v\delta_\iota) ww'}_{i'}$.
\end{observation}

Observations \ref{obs:easy1}-\ref{obs:easy3}, with the  Lemma \ref{lem:prettier2},
imply that there is a homomorphism from $\adr{\theforest_{v\delta_\iota}}$  to $\adr{\theforest_{\btr(v\delta_\iota)}}$.
Now we want something more, we want to prove that there is a homomorphism from
$\adr{\theforest_{v\delta_1}\cup \theforest_{v\delta_2}}$ to 
$\adr{\theforest_{\btr(v\delta_1)}\cup \theforest_{\btr(v\delta_2)}}$. 

\begin{lemma}\label{lemma:easy4} 
Let $w,w'\in {\mathbb A}^*$. If 
$\adr{v\delta_1 w}_i= t= \adr{v\delta_2 w'}_{i'}$ for some term $t$. Then also
$\adr{\btr(v\delta_1) w}_i= \adr{\btr(v\delta_2) w'}_{i'}$.
\end{lemma}

\begin{proof}
From Observation \ref{lem:local-terms-belong-to-a-border-trigger} we know that 
$t=\adr{v\delta_1}_j=\adr{v\delta_2}_{j'}$ for some positions $j,j'$. Now recall that 
the pair $\bt{\btr(v\delta_1),\btr(v\delta_2)}$ is a blocking team for $(\delta,v)$. By condition (iii) of Observation 
\ref{obs:trivial} this gives us that $\adr{\btr(v\delta_1)}_j= \adr{\btr(v\delta_2)}_{j'}$. Now use Observation
\ref{obs:easy3}.
\end{proof}

\subsection{Proof of  Lemma \ref{lem:crucial}. The Harder Part}

For ease of reading, we restate \cref{obs:trivial} here:%
\begin{customobs}{\ref{obs:trivial}}[shortened]
Let $\rho$ be a rule of the form
\centerline{ $\Forall{\vx \vy} \phi(\vx, \vy) \to \Exists{\vz} \psi_1(\vy, \vz)\wedge \psi_2(\vy, \vz).$} \\Then $\bt{v_1,v_2}$ is a blocking team for trigger $\pair{\rho, u}$ \iffi
\begin{enumerate}[(i)]
    \item the predicate symbols of  $\adr{v_\iota}$ and of $\adr{u\rho_\iota}$ are equal;
    \item if $j$ is a frontier position in $\psi_\iota$ then $\adr{u\rho_\iota}_j= \adr{v_\iota}_j$;
    \item if $j$ is any  position in $\psi_\iota$ and $j'$ is any  position in $\psi_{\iota'}$ and $\adr{u\rho_\iota}_j=\adr{u\rho_{\iota'}}_{j'}$ then also $\adr{v_\iota}_j=\adr{v_{\iota'}}_{j'}$.
\end{enumerate}
\end{customobs}

Now we finish the proof of \cref{lem:crucial}.
%
%
Let $\rho$ be as in \cref{obs:trivial}.
Since $\pi$ is a trigger from $\derriv$, all terms which occur in atoms $\adr{u}$, $\adr{u\rho_1}$ and $\adr{u\rho_2}$ are global.
Since $\bt{w_1,w_2}$ is a blocking team for $\pi$, we know that it satisfies the conditions {(i)--(iii)} from Observation \ref{obs:trivial}.
Our goal is to prove that $ \bt{\presuperior(w_1),  \presuperior(w_2)}$ also satisfies these {three} conditions.

\noindent
\textbf{Condition (i).}
{We know, from} Lemma \ref{lem:prettier2}, that each of $\presuperior(w_\iota)$ is  prettier than $w_\iota$.
This means, in particular, that $\presuperior(w_\iota)$ and $w_\iota$ have the same predicate symbol, and thus
condition (i) from Observation \ref{obs:trivial} is satisfied for $\presuperior(w_\iota)$.

\noindent
\textbf{Condition (ii).}
Consider some frontier position $j$ in $\psi_\iota$. We know, from Observation \ref{obs:trivial} (ii) that
$\adr{u\rho_\iota}_j= \adr{w_\iota}_j$. Since $\adr{u\rho_\iota}_j$ is a global term, and since $w_\iota \lesspretty \presuperior(w_\iota)$,
we get that condition (ii) from this Observation  also holds for $\presuperior(w_\iota)$.

\noindent
\textbf{Condition (iii).}
Here things are more complicated. Directly from the fact that $w_\iota \lesspretty \presuperior(w_\iota)$ we get that condition (iii)
is satisfied for $ \bt{\presuperior(w_1),  \presuperior(w_2)}$, but only if $\iota=\iota'$.

Recall that 
$\bt{w_1,w_2}$ is a blocking team, so it satisfies condition (iii) for $\rho$. So, 
in order to  prove that (iii) is also true for 
 $ \bt{\presuperior(w_1),  \presuperior(w_2)}$ 
it suffices to show that 
if $\adr{w_1}_j=t=\adr{w_2}_{j'}$ for some term $t$ then also 
$\adr{\presuperior(w_1)}_j=\adr{\presuperior(w_2)}_{j'}$.




Now, if $t$  is global, then
 Lemma \ref{lem:prettier2} and Def. \ref{lem:def-presuperior} tell us that $t=\adr{w_1}_j =  {\adr{\presuperior(w_1)}_j}$ and
$t=\adr{w_2}_{j'}= {\adr{\presuperior(w_2)}_{j'}}$.

If $t$ is a local term, then
from Observation \ref{lem:local-terms-belong-to-a-border-trigger} (iii) we learn that 
there exists trigger $(\delta,v)$ such that 
$w_1,w_2\in \theforest_{v\delta_1}\cup  \theforest_{v\delta_2}$. Now 
use Observation \ref{obs:easy3} and Lemma \ref{lemma:easy4}.\qed

\subsection{From Corollary \ref{lem:so-useful} to the First Main Lemma}\label{sec:creating-mixed}

In this subsection we  construct $\demriv$ such that  $\pair{\demriv, \derriv}$ is a mixed derivation.
Define $\demriv$ as any oblivious derivation such that:
\begin{itemize}
  \item $\derriv$ is a sub-sequence of $\demriv$, that is $\demriv$ respects the order of triggers in $\derriv$.
  \item If $\pi=(\rho,u)$ is a trigger such that $\pi \not \in \derriv$, and if $\mu_1$ and $\mu_2$ are triggers in 
  $\derriv$, which produce 
  $\btr(u\rho_1)$ and  $\btr(u\rho_1)$, then $\pi$ appears in $\demriv$ after both $\mu_1$ and $\mu_2$.
   
  \end{itemize}

Note that each trigger that does not appear in $\derriv$ only needs to wait for at most two triggers from
$\demriv$ to get its slot in $\demriv$. Therefore such a derivation $\demriv$ exists.
It also follows easily from the construction, and from \cref{lem:so-useful}, that $\pair{\demriv, \derriv}$ is indeed a mixed derivation.


%% file: 03b-mixed-to-atomic-path.tex
\newcommand{\theforestomega}{\theforest_{\singleton{\omega}}}
\section{\texorpdfstring{$\omega$}{omega}-Sensitive-Path-Derivations}\label{sec:atomic}

We now introduce yet another kind of derivation:

\begin{definition}

An \emph{$\omega$-sensitive-path-derivation}
is a pair $\pair{\demriv, \depriv}$ where $\demriv = \{\pi^\demriv_n\}_{n\in \mathbb N}$  is an infinite oblivious derivation over $\{\omega\}$, and
$\depriv$ is an infinite restricted derivation over $\{\omega\}$, such that:

\begin{enumerate}[(i)]
\item $ \depriv$ is a subsequence of $\demriv$;
\item if $\pi^\demriv_n\in \depriv$ for some $n$ then $\pi^\demriv_n$ is not blocked in  $\demrivstep{n}$;
\item the set $\{u:\; \exists \rho\in \rs \;(\rho,u)\in \depriv\} $ is an infinite path.
\end{enumerate}

Since we adopted the convention that the symbol $\deprivforest$ is kept for 
the set of all addresses created by $ \depriv$, we will use $\ppath$ to denote  the infinite path
from  (iii). It is easy to see that $\deprivforest^{\text{\tiny 0}}\subset \deprivforest$,
and that $w\rho_\iota\in \deprivforest$ if and only if $w\rho_{\iota'}\in \ppath$ for some $\iota'$.

\end{definition}

The above definition is similar to Definition \ref{def:mixed-chase} of Mixed Derivation, but with three differences: 
unlike the $\derriv $ in Definition \ref{def:mixed-chase}, our $\depriv$  {\bf does not} need to be fair.
On the other hand, now we require that the set of the addresses of triggers in $\derriv $ (which we see as ``sensitive'', as opposed to the oblivious triggers in $\demriv\setminus\derriv$ ) is not some 
arbitrary forest, like in Definition \ref{def:mixed-chase}, but something very specific, namely a path.
The third difference is that $\omega$-sensitive-path-derivation does not start from some arbitrary 
instance $\inst$ but from a single atom. 

The following is our Second Main Lemma.

\begin{lemma}[Second Main Lemma]\label{lem:second-main-lemma}
The following two statements are equivalent:

\begin{enumerate} 
  \item There exists a mixed derivation for some $\inst$.
  \item There exists an $\omega$-sensitive-path-derivation for some $\omega$.
\end{enumerate} 
\end{lemma}
\begin{proof}
    The proof can be found in appendix \cref{app:second-main}.
\end{proof}

Notice that the expression \emph{``$\omega$-sensitive-path-derivation''} 
in the lemma could not 
be replaced with \emph{``an infinite restricted (but possibly not fair) chase derivation creating an infinite path.''} That is, we cannot skip the conditions relating to $\demriv$. This is because, as 
Example \ref{ex:drugi} shows, such an infinite unfair derivation can easily exist even if $\rs \in \allallterm$.

Lemma \ref{lem:second-main-lemma} looks inconspicuous, but it constitutes a major step in our reduction from the problem $\allallterm$ to Monadic Second Order Logic. And let us try to explain why.

$\theforest$ is  (almost) a tree and the steps of a restricted derivation are its subtrees. So, one could think, nothing is more natural than the idea of employing MSOL, which is decidable on such trees, as a decision procedure for termination. 

It is not that simple. Restricted chase is about the ordering of trigger applications, and it is not at all clear how this ordering could be expressed in MSOL. One could naively imagine the problem could be solved by a formula like \emph{``For each $u\in \theforest$ there exists a set $\agraph$ of all addresses that were created before $u$''}. But this would not work, as there is no way to enforce consistency between such sets $\agraph$, that is to make sure that the sets produced by such universal quantifiers  faithfully represent steps of the same derivation.

\cref{lem:second-main-lemma} (and \cref{trm:restrictedChaseToMixedChase}) says that one can reduce $\allallterm$ to the existence of a sort of infinite restricted chase, in which (due to its linear nature) the order of applications is undisputed and does not need to be guessed by our formula.  

Of course, $\omega$-sensitive-path-derivations feature oblivious triggers which somehow need to interact with the sensitive ones
but, as we will see in Lemma~\ref{th:mixed-to-properties-on-thetree}, they are manageable.

%% file: 04-mso.tex
\newcommand{\existsSensitivePath}{\mathsf{\Omega}}
\newcommand{\isAbsRep}{\ensuremath{\mathsf{correct}}\xspace}

\section{Employing Monadic Second Order Logic}\label{sec:mso}
In this last section of this paper we prove:

\begin{lemma}[Third Main Lemma]\label{lem:third-main} One can construct, for a given atom $\omega$,
 a sentence $\existsSensitivePath_{\omega}$ 
of MSOL such that $\omega\addresses^* \models \existsSensitivePath_{\omega}$ \iffi there exists an $\omega$-sensitive-path-derivation.
\end{lemma}

Notice that, in view of Lemmas \ref{lem:second-main-lemma} and 
\ref{trm:restrictedChaseToMixedChase},
 and due to decidability of Monadic Second Order Logic on infinite trees, 
this will end the proof of Theorem \ref{th:main}.\smallskip

For the rest of this paper we set  $\inst = \{\omega\}$, for a fixed $\omega$, so $\theforest\subseteq \omega\addresses^*$.
Let us now characterize existence of 
$\omega$-sensitive-path-derivations in a language better suited for MSOL.\smallskip


\noindent
{Let $\ppath$ be a path in $\theforest$, and (using notations from  Section \ref{sec:atomic})} let 
$\depriv
$ be the (unique!) infinite derivation comprising all the triggers needed to create $\ppath$. Recall that the set
$\deprivforest$ of addresses created by $\depriv$ contains, together
with each address $u\rho_\iota$ in $\ppath$ also the address $u\rho_{\iota'}$ for $\iota'\neq \iota$.
Notice also, that for each $n\in \mathbb N$ there is exactly one $u\in\ppath$ such that $|u|=n$.

\begin{definition} 

Given an address $u \in \ppath$  we define:
\begin{itemize}
  \item
 $\preSet(u)$ as the set of  addresses  $v\in\deprivforest$ such that 
 $|v|\leq |u|$;
  
  \item $\freeSet(u)$ as the set of addresses  $v\not\in\deprivforest$ without $u$ as prefix.
  \end{itemize}
\end{definition}

\noindent
The idea is that $\preSet(u)$ are the addresses {that must be created \emph{before} $u$ in every oblivious derivation},
 and $\freeSet(u)$ are the addresses that can be created at {\emph{any}} time, before or after $u$.

For $u\in \ppath$ by $\pi_u$ denote the (unique) trigger of the form $(\rho,u)$ in $\depriv$.
Also, $\mathsf{team}(v,v',\pi)$ is short for ``$\bt{v,v'}$ is a blocking team for $\pi$.''
For $u\!\in\! \ppath$ define $\fragileSet(u)$ as the set: $$\singleton{w\!\in\!\ppath \!: u\!\in\!\preSet(w)\,\wedge\, \thinexists{v, s \!\in\! \freeSet(u)\!\cup\! \preSet(w)} \; \mathsf{team}(v,s,\pi_w)} $$


\begin{lemma}\label{th:mixed-to-properties-on-thetree}
The following statements are equivalent:
\begin{enumerate}
  \item There exists an $\omega$-sensitive-path derivation
  \item 
There exists an infinite  path $\ppath\subseteq \theforest$ s.t. for each  $u \in \ppath$:
\begin{enumerate}[(i)]
\item 
\!$\pi_u$ is not blocked by any team $ \bt{v,\!v'}$ for  $v,\!v'\!\!\in\! \preSet(u)$.
\item 
\!The set $\fragileSet(u)$  is finite.
\end{enumerate}
\end{enumerate}
\end{lemma}

Before we prove Lemma \ref{th:mixed-to-properties-on-thetree}, 
let us explain how it implies Lemma \ref{lem:third-main}. 
And this is actually by fairly routine reasoning.


\noindent
{{\bf From \cref{th:mixed-to-properties-on-thetree} to \cref{lem:third-main}.}}
{Recall,} that all the MSOL formula $\existsSensitivePath_{\omega}$  will be able to see
is $\omega\addresses^*$, and the relations {$\child{\rho_\iota}$, for each rule $\rho$ and $\iota$}. Using this language we want to write a formula
equivalent to Statement 2 of 
\cref{th:mixed-to-properties-on-thetree}.

 So, first of all, we need the formula to be able to guess, and verify (using monadic second order resources), which addresses {$w$} are really in $\theforest$ and, if $w\in \theforest$, what is  $\adr{w}$. There is a little problem here, 
since there are infinitely many atoms in $\adr{\theforest}$ and we can only afford to have finitely many unary predicates in $\existsSensitivePath_{\omega}$. This is solved by guessing a function $\absadr{\_}$, which is a ``compressed'' version of
$\adr{\_}$. {This guess can be made - using standard techniques - as the treewidth of $\thechase$ is less than the arity of ${\mathfrak P}$}.

Then, having defined this $\absadr{\_}$, we are able to write a formula\footnote{Strictly speaking, we have one such formula for each pair $i,j$ for $i$ and $j$ not greater {than the arity of $\mathfrak P$.}} $\simeq_{i.j}$, with two free first order 
variables, such that $\omega\addresses^* \models u \simeq_{i.j} v$ if and only if $\adr{u}_i=\adr{v}_j$. Observation \ref{lem:local-terms-belong-to-a-border-trigger}(i) and (ii) is helpful here.

{With $\simeq_{i.j}$ at hand, writing $\existsSensitivePath_{\omega}$ becomes a classroom exercise in MSOL programming. Indeed, the definitions of the notions of trigger and blocking team, as presented in Section \ref{sec:triggers}, rely solely on equalities between terms - which we are now able to express. What then remains to be expressed in Statement 2 of Lemma \ref{th:mixed-to-properties-on-thetree}, such as infinite paths or infinite sets, are part of the standard MSOL toolkit.}
{The full argument can be found in appendix \cref{app:mso-extended}.}

\noindent

Now, in order to finish the proof of Lemma \ref{lem:third-main}, and hence also Theorem \ref{th:main}, we only need
to prove
Lemma  \ref{th:mixed-to-properties-on-thetree}.

%% file: 06-magic-proof.tex
\subsection{Proof of Lemma \ref{th:mixed-to-properties-on-thetree}}\label{sec:magic-proof}

\begin{figure}
    \centering
    \includegraphics[width=0.45\linewidth]{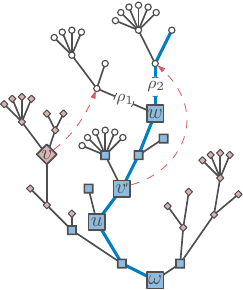}
    \caption{A visual guide to condition (ii) of Statement 2 from \cref{th:mixed-to-properties-on-thetree}. Node $w$ appears in $\fragileSet(u)$ and trigger $\pi_w = \pair{\rho, w}$ is blocked by a team of nodes $\bt{v,v'}$ represented with red dashed edges. Bold-blue path represents $\ppath$, blue square nodes are $\preSet(w)$, red diamond nodes are $\freeSet(u)$.}
    \label{fig:magic}
\end{figure}

Let $\ppath$ be an arbitrary infinite path in $\theforest$, 
and let
$\depriv $ relate to 
$\ppath$ in the usual way.
In order to show Lemma \ref{th:mixed-to-properties-on-thetree} it is enough to
prove that the following two conditions are equivalent:

\begin{enumerate}
\item[A.] There exists an oblivious derivation $\demriv$
 such that  $\pair{\demriv, \depriv}$ is an $\omega$-sensitive-path-derivation.
\item[B.] $\ppath$ satisfies  conditions (i) and (ii) from Lemma \ref{th:mixed-to-properties-on-thetree} (2).\smallskip
\end{enumerate}

\noindent
{\bf B $\Rightarrow$ A.}
Can be found in appendix \cref{app:lemma39}.

\noindent
{\bf A $\Rightarrow$ B.} Now suppose $\pair{\demriv, \depriv}$ is an $\omega$-sensitive-path-derivation, for some oblivious derivation $\demriv$.

First of all, notice that condition (i) is trivially satisfied for each $u$, since for a pair 
 $\pair{\demriv, \depriv}$ to be an $\omega$-sensitive-path-derivation, $\depriv$ itself must be 
 a restricted derivation. 

We now focus on condition (ii). Let us fix some 
$u\in\ppath$. We need to show that the set  $\fragileSet(u)$ is finite.
The  idea is  that $u$  ``controls" the communication between elements of $\freeSet(u)\cup \preSet(u)$ and elements $w$ in $\ppath$ for 
which
$u\in \preSet(w)$. 

Let us now try to formalize this idea. 
Suppose $\pi_u=(\rho,u)$, which means that  rule 
$\rho$ is the one that, applied to $\adr{u}$, produced the
the element following $u$ on $\ppath$. Suppose also that $\rho$ is:\hfill
$\Forall{\vx \vy} (\phi(\vx, \vy) \to \Exists{\vz} \psi_1(\vy, \vz)\wedge \psi_2(\vy, \vz))$.

Then the following is of course true:

\begin{observation}\label{obs:komunikacja}
If a term $t$ occurs in an atom of $\adr{\freeSet(u)\cup \preSet(u)}$ and in some 
$\adr{w}$ such that 
$w\in\ppath$ and
$u\in \preSet(w)$, then there is a frontier position $i$ in $\phi$ such that 
$t=\adr{u}_i$.

\end{observation}

Now, imagine you are the atom $\adr{u}$ and you want to somehow classify the atoms and the pairs in 
$\adr{\freeSet(u)\cup \preSet(u)}$.
 This is what you would do:

\begin{definition}\label{def:sim}
For  addresses $v,v'\in \freeSet(u)\cup \preSet(u) $ we define $v\sim v'$ if and only if
there is an isomorphism, from $\adr{v}$ to $\adr{v'}$, which is an identity on terms 
which occur in $\adr{u}$.
\end{definition}

\begin{definition}
For  addresses $v,v',s,s'\in \freeSet(u)\cup \preSet(u) $ we define ${\pair{v,s}\approx \pair{v',s'}}$ \iffi
there is an isomorphism, from {$\pair{\adr{v},\adr{s}}$ to $\pair{\adr{v'},\adr{s'}}$}, which is an identity on terms 
in $\adr{u}$.
\end{definition}

Both $\sim$ and $\approx$ are  equivalence relations with finite index.
Let now $w\in \ppath$ be such that $u\in \preSet(w)$.
The next  two observations follow easily from Observations \ref{obs:komunikacja} and  \ref{obs:trivial}:

\begin{observation}\label{obs:takie-same-2}
Suppose $v,v',s,s'\in \freeSet(u)\cup \preSet(u) $ and {$\pair{v,s}\approx \pair{v',s'}$}.
Suppose also that $\bt{v,s}$ is a blocking team for $\pi_w$.
Then $\bt{v',s'}$ is also a blocking team for $\pi_w$.
\end{observation}

\begin{observation}\label{obs:takie-same-1}
Let $v,v'\in \freeSet(u)\cup \preSet(u) $ and ${v}\sim {v'}$.
Let also $s\hspace{-1mm}\in\hspace{-1mm}\deprivforest$ be such that 
$|u|\hspace{-0.5mm}\leq|s|\hspace{-0.5mm}\leq|w|$ and $\bt{v,s}$ is a blocking team for $\pi_w$.
\hspace{-2mm} Then $\bt{v',s}$ is  a blocking team for $\pi_w$.
\end{observation}

  The above holds because whenever $\adr{v}_i = \adr{s}_j$ for some positions $i$ and $j$, we have that $\adr{v}_i$ is a term of $\adr{u}$ and hence 
  (by Definition \ref{def:sim}) we get $\adr{v}_i=\adr{v'}_i$.

We are going to present the set $\fragileSet(u)$ as a certain union of sets.
But first we need:

\begin{definition}
Suppose $v,s\in \freeSet(u)\cup \preSet(u) $. Then:\smallskip\\
$ \harmb(v,s)=\{w\in \ppath : u\in \preSet(w) \;\wedge \;  \teamtriple(v,s,\pi_w)\}  $\smallskip\\
$ \harma(v)=$\smallskip\\
\mbox{$\{w\in \ppath : u\in \preSet(w) \; \wedge \; \exists_{s\in \preSet(w)\setminus \preSet(u)}  \;\teamtriple(v,s,\pi_w)\}  $} 
\end{definition}

Our next observation (still assuming that $\pair{\demriv, \depriv}$ is an $\omega$-sensitive-path-derivation) is that:

\begin{lemma}\label{lem:skonczone}
The sets $\harma(v)$ and $\harmb(v,s)$ are finite for all $v,s\in \freeSet(u)\cup \preSet(u)$.
\end{lemma}
\begin{proof} 
Proof can be found in appendix \cref{app:mso-extended}.
\end{proof}


It follows directly from the definition of $\fragileSet(u)$ that:

\begin{observation}\label{obs:suma}
The set $\fragileSet(u)$ is equal to:\medskip\\
$\bigcup_{v,s\in \freeSet(u)\cup \preSet(u)} \harmb(v,s)\cup
\bigcup_{v\in \freeSet(u)\cup \preSet(u)} \harma(v)$
\end{observation}

Using Observations \ref{obs:takie-same-2} and \ref{obs:takie-same-1}
it is also easy  to see that:

\begin{observation}\label{obs:ostatnia}
If $v \sim v'$ then    $ \harma(v)=\harma(v').$\\
If $\pair{v,s}\approx\pair{v',s'}$ then    $ \harmb(v,s)=\harmb(v',s')$.
\end{observation}

Now, since $\sim$ and $\approx$ are equivalence relations of finite index, it follows from 
Observation \ref{obs:ostatnia} that
the union of sets in Observation \ref{obs:suma} is in fact a finite union. 
Together with Lemma \ref{lem:skonczone}, this means, that $\fragileSet(u)$ is a finite union of finite sets and hence a finite set. 

%% file: 07-conclusion.tex
\section{Conclusions and Future Work}

The next natural step is the decidability of the same problem for multi-headed guarded rules.

\begin{definition}
	An existential rule is \emph{guarded} if its body contains an atom, the \emph{guard}, featuring all variables of the body.
\end{definition}

\begin{conjecture}
	Membership in $\allallterm$ for sets of multi-headed guarded existential rules is decidable.
\end{conjecture}

However, it is still unclear how the techniques used for the linear case extend to the guarded case, since they rely heavily on the forest structure of the chase. 

One could still define a forest structure in the guarded case, by setting $\adr{w\rho_\iota}$ as the $\iota$-th atom resulting from the application of rule $\rho$ with guard $\adr{w}$, but this induces other problems. Indeed, the creation of $\adr{w\rho_\iota}$ requires some side-atoms beside the guard, which may have addresses far from $w$. Consider for instance the guarded rule
\[E(x,y),A(y)\to\Exists{z} E(y,z)\]
and suppose that some vertices $v$ and $w$ are labelled with $A(b)$ and $E(a,b)$, respectively. Then, $w$ has a child $wa$ labelled $E(b,c)$, but any chase sequence must create $v$ before $wa$. We thus need to enforce a total order on distant elements of the forest, which cannot be done by MSOL means in any obvious way.

MSOL techniques have already been successful at dealing with the single-headed guarded case \cite{GMP23}, but in this case the tree considered is a join tree, which always exists in the single-headed case but not in the multi-headed case.
A natural generalization of join trees would be tree decompositions, which could be used since the chase with guarded rules admit finite treewidth. However, it is even less clear how the techniques developed for the linear case could apply on a tree decomposition, especially taking into account the fact that nodes are now labelled with sets of atoms and not just single atoms, and that this does not solve the issue of having to enforce an order on distant elements of the forest.

%% file: A-second-main.tex
\section{Proof of Lemma \ref{lem:second-main-lemma}}\label{app:second-main}

\begin{customlem}{\ref{lem:second-main-lemma}}
The following two statements are equivalent:
\begin{enumerate} 
  \item There exists a Mixed Derivation for some $\inst$.
  \item There exists an $\omega$-sensitive-path-derivation for some $\omega$.
\end{enumerate} 
\end{customlem}

 
\noindent
{\em Proof (1 $\Rightarrow$ 2).}
Assume $\pair{\demriv, \derriv}$ is a Mixed Derivation over some instance $\inst$. Take an infinite path $\flower$ of $\derrivforest$ starting in some atom $\omega$ of $\inst$, where $\derrivforest$ is, as usual, the set of addresses created by $\derriv$. Let $\depriv$ be a restricted derivation,
which is a subsequence of $\demriv$, and
such that for each $\pair{\rho, u} \in \depriv$ we have that $u\rho_\iota \in \flower$, for some $\iota$. 

Let now $\derivp$ be the derivation which is a subsequence of  $\demriv$, consisting of all triggers of $\demriv$ of the form $\pair{\delta, \omega u}$ for some $u$. Then $\derivp$ is an oblivious derivation over $\{\omega\}$.

It is now easy to notice  that $\pair{\derivp, \depriv}$ is indeed an $\omega$-sensitive-path-derivation. Consider a trigger $\pi$ of $\depriv$. Let $n$ be its position in $\derivp$ and $m$ be its position in $\demriv$. We know that $\pi$ is not blocked in the $m$-th step $\demrivstep{m}$ of $\demriv$. And, since the $n$-th step $\derivpstep{n}$of $\derivp$ is a subset of $\demrivstep{m}$, we have that $\pi$ is not blocked in $\derivpstep{n}$.\smallskip

\noindent
{\em Proof (2 $\Rightarrow$ 1).}
Assume $\pair{\demriv, \depriv}$ is an $\omega$-sensitive-path-derivation,
for some oblivious derivation $\demriv=\singleton{\derivtrig{n}}_\nN$ over $\singleton{\omega}$. 
%
%
Our intention is to build an infinite restricted and fair derivation $\derriv$ over $\{\omega\}$, being a
subsequence of $\demriv$, and such that $\pair{\demriv, \derriv}$ is a Mixed Derivation. 

In the following, we construct a derivation $\derriv$ by iterating over elements of $\demriv$. The construction is straightforward: 
Suppose $\derriv_n$ is the already constructed prefix of $\derriv$, and $\derrivforest^{n}$ is the set\footnote{Notice that $\derrivforest^{n}$ is not the $\derrivstep{n}$ from Definition \ref{def:chase-derivations}.}  containing $\omega$ and all addresses created by triggers in $\derriv_n$. If $\derivtrig{n}$ is in 
$\derrivforest^{n}$ and is
not blocked in $\derrivforest^{n}$, 
then define $\derriv_{n + 1}$ as the sequence $ \derriv_n; \derivtrig{n}$. Otherwise $\derriv_{n + 1} = \derriv_n$. 
Finally, define $\derriv$  as the union\footnote{Or, to be more precise, as the minimal sequence having each $\derriv_n$ as a prefix. } of all  $\derriv_n$. 

\smallskip
\noindent
The implication (2 $\Rightarrow$ 1) follows from the next two lemmas.%
\vspace{-0.5mm}%
\begin{lemma}\label{obs:mixed-restriction-is-restricted-and-fair}
Derivation $\derriv$ is restricted and fair.
\end{lemma}

\begin{proof}
Derivation $\derriv$ is restricted by definition. Now, we show it is fair.
To this end we need to show, for each trigger $\derivtrig{n} \in \demriv$, that if it 
appears in  $\derrivforest$ then it is either in $\derriv$ or is blocked in $\derrivforest$.
Consider some trigger $\derivtrig{n}=(\rho,u)$.
Clearly, the trigger that creates $u$ is somewhere before $\derivtrig{n}$ in $\demriv$.
Hence, if $\derivtrig{n}$ appears in $\derrivforest$, then it must appear already in 
$\derrivforest^{n}$. But if $\derivtrig{n}$ appears in $\derrivforest^n$ and is not blocked 
in $\derrivforest^n$ (and, in consequence, also in  $\derrivforest$) then it is in $\derriv_{n + 1}$. 
\end{proof}

\begin{lemma}\label{lem:derivp-contains-S}
Derivation $\derriv$ is infinite. 
\end{lemma}
\begin{proof} It will be enough to show that  $\derriv$ contains all triggers from $\defriv$.
Suppose, towards contradiction, that $\derivtrig{n}=(\rho,u)$ is the first trigger of $\defriv$ which is 
not in $\derriv$. Since all the previous triggers of $\defriv$ are in $\derriv$,
we know that $u\in \derrivforest^n$. So, the only reason for $\derivtrig{n}$ to be not in $\derriv$
is that it is blocked in $\derrivforest^n$.
But this contradicts the assumption that
$\pair{\demriv, \depriv}$ is an $\omega$-sensitive-path-derivation
\end{proof}

%% file: B-space-saver.tex
\section{Proof of Lemma \ref{th:mixed-to-properties-on-thetree}; {\bf B $\Rightarrow$ A.}}\label{app:lemma39}

We need to somehow construct $\demriv$ such that $\pair{\demriv, \depriv}$ is an $\omega$-sensitive-path-derivation. This means that, like in Section \ref{sec:mixed}, we need to slot the triggers not in $\depriv$ between triggers from  $\depriv$ in a way that does not block triggers in  $\depriv$.

By assumption, for each $u\in\ppath$, the set $\fragileSet(u)$
is finite. So we can define a function $\safe: \ppath \rightarrow \ppath$ which provides, for each 
 $u\in\ppath$, an address $\safe(u)\in\ppath$ which is above all elements of $\fragileSet(u)$.
Formally speaking, $\safe(u)$ is any element of $\ppath$ such that $|u|<|\safe(u)|$ and for each $w\in \fragileSet(u)$ it holds that $|w|<|\safe(u)|$.

Notice that, for each trigger $\pi=(\rho, v)\notin\depriv$,
there exists $u\in \ppath$ such that $v\in \freeSet(u)\cup \preSet(u)$.
We can then safely slot $\pi$ anywhere 
after the trigger creating $\safe(u)$.\smallskip

%% file: 05A-MSO-formulae.tex
\section{Extended arguments for Section \ref{sec:mso}}\label{app:mso-extended}

\subsection{From \texorpdfstring{$\omega\addresses^*$}{addresses} to \texorpdfstring{$\theforest$}{F}}\label{ssec:abstract-representations}

As all $\existsSensitivePath_{\omega}$  can see is $\omega\addresses^*$, we need a way for this formula to guess, and verify (using monadic second order resources), which addresses belong to $\theforest$
and, for each address $w\in\theforest$, what is the value of $\adr{w}$. This is a bit tricky, 
since there are infinitely many atoms in $\adr{\theforest}$ and we only can afford to have finitely many unary predicates in $\existsSensitivePath_{\omega}$. But we can guess a ``compressed'' version of
$\adr{\_}$ instead:

\begin{definition}
We shall call constants of $\omega$ and natural numbers each not greater than {thrice} the maximal arity of predicates from ${\mathfrak P}$ (denoted as $\aamaxint$) \emph{abstract terms}.
  An \emph{abstract atom} $\apred(\vk)$ is an atom such that $\apred\in {\mathfrak P}$, and $\vk$ is a tuple of abstract terms. 
  We denote by $\abstractatoms$ the set containing all abstract atoms and an additional symbol $\bot$.
\end{definition}

The set  $\abstractatoms$ is, of course, finite, so  we will be able to think of its elements as names of 
monadic predicates. 

The way the next definition is formulated may seem weird, but this is how 
Monadic Second Order Logic wants it to be written:

\begin{definition}\label{def:absadr}
Consider a relation ${\mathfrak R}\subseteq \omega{\mathbb A}^* \times \abstractatoms $, {denoted as $\absadr{\_}$ when ${\mathfrak R}$ is a function.
${\mathfrak R}$  is called \emph{correct} if:}
\begin{itemize}
\item it is {indeed} a function,
  \item  {$\adr{\omega}$ is\footnote{{Note that $\adr{\omega}$ contains only constant terms and so $\adr{\omega}$ is an abstract atom itself.}}  $\absadr{\omega}$,}
  \item for each $u\in \theforest$ and a trigger $\pair{\rho, u}$ we have that the tuple $\big(\adr{u},\adr{u\rho_1},\adr{u\rho_2}\big)$ is isomorphic to the tuple $\big(\absadr{u}, \absadr{u\rho_1}, \absadr{u\rho_2}\big)${, and}
  \item for each $u\in \omega{\mathbb A}^*\setminus \theforest$ we have  $\absadr{u}=\bot$.
\end{itemize}

We also use the term \emph{correct} for {$\absadr{\_}$ when $\mathfrak{R}$ is a function.}
\end{definition}

Intuitively, for a correct ${\mathfrak R}$, the set  $\theforest$ with function $\absadr{\_}$
looks locally exactly as $\theforest$ with function $\adr{\hspace{-0.15mm}\_\hspace{-0.15mm}}$\hspace{-0.3mm}. It is {straightforward to} create\footnote{This is because the tree width of the infinite instance $\adr{\theforest}$ is smaller than $\aamaxint$.}  a correct $\absadr{\_}$
 from $\adr{\_}$ carefully replacing nulls in atoms of $\adr{\theforest}$ with 
 {naturals} not bigger than $\aamaxint$.

Now, let us think about a relation ${\mathfrak R}\subseteq \omega{\mathbb A}^* \times \abstractatoms $.
Such relation assigns, to each element of $\omega{\mathbb A}^*$, a subset of the finite set  $\abstractatoms $. Recall however that we see elements of  $\abstractatoms$ as monadic predicates.
So ${\mathfrak R}$ 
can be seen as a tuple of  monadic predicate symbols. 

Now it is quite easy (albeit certainly not pleasant) to write a first order logic formula $\isAbsRep$,
whose signature consists of relations $\child{a}$ for $a\in \mathbb A$ (recall Definition \ref{def:address}) and of 
the monadic relation symbols from $\abstractatoms$, such that:
$$\omega{\mathbb A}^*\models \isAbsRep({\mathfrak R}) \;\; \;\; \text{if and only if}\;\;\;\; {\mathfrak R\;}\text{is correct}  $$

\noindent
Clearly, formula $\isAbsRep$ {depends on}  $\rs$ and the atom $\omega$.

\subsection{Defining \texorpdfstring{$\existsSensitivePath_{\omega}$}{the formula}}

It is hardly surprising  that formula  $\existsSensitivePath_{\omega}$ will have the form:
$$\exists \mathfrak{R}, \ppath \;\;\;\isAbsRep(\mathfrak{R}) \wedge \isFlower(\ppath) \wedge 
\existsSensitivePath^0_{\omega}$$
Where  $\existsSensitivePath^0_{\omega}$ will be a formula expressing that $\ppath$ satisfies 
statement 2 from Lemma \ref{th:mixed-to-properties-on-thetree}. Both 
$\existsSensitivePath^0_{\omega}$ and $\isFlower$
will 
use  function $\absadr{\_}$ as part of their vocabulary, assuming that it is correct. {Notice that also $\theforest$ is now definable, by a formula $\absadr{u} \neq \bot$, so we can have it in the vocabulary.
} Also, $\existsSensitivePath^0_{\omega}$ will use the unary predicate $\ppath$ as part of its vocabulary, assuming that it indeed represents an infinite path in $\theforest$.

In the following we will make use of the well known: 

\begin{observation}\label{form:infrastructure}
There exist  MSOL formulae:
\begin{itemize}
  \item $\isFlower(P)$ which holds \iffi  $P\subseteq\theforest$ is an infinite path;
  \item $\infinite(R)$ which holds \iffi $R$ is an infinite subset of $\theforest$;
  \item $u \earliereq w$ which holds \iffi $u$ is a prefix of $w$;
  \item $\ancestralPath(P, u, w)$ which holds \iffi  $u \earliereq w$ and $P$ is equal to $\singleton{v \in \theforest \mid u \!\earliereq\! v \!\earliereq\! w }$.
 \end{itemize}
\end{observation}

In order to express trigger blocking using an MSOL formula, we need to be able to say, for 
$u,v\in \theforest$, that $\adr{u}_i=\adr{v}_j$. Recall however, that we no longer have access to 
$\adr{\_}$, and we have $\absadr{\_}$ instead.\smallskip

\noindent
For each $i,j$ not greater than the maximal arity of predicates in $\mathfrak P$
   we define a binary relation $\simeq^0_{i,j}$  over elements of $\theforest$:
   
\begin{definition} 
    $u\rho_\iota \simeq^0_{i,j} u\rho_{\iota'}$ \;\iffi\; $\absadr{u\rho_\iota}_i = \absadr{u\rho_{\iota'}}_j$.
\end{definition}

Obviously, each $\simeq^0_{i,j}$ is first order definable, so our MSOL formulae can use it as part of their vocabulary. And (recall the third bullet of \cref{def:absadr}), 
$\absadr{u\rho_\iota}_i = \absadr{u\rho_{\iota'}}_j$ 
\iffi
$\adr{u\rho_\iota}_i = \adr{u\rho_{\iota'}}_j$, so formula $\simeq^0_{i,j}$ indeed expresses 
equality of respective terms in $\adr{u\rho_\iota}$ and $ \adr{u\rho_{\iota'}}$
\smallskip

Now, as our next step, we will define, again for 
each $i,j$ not greater than the maximal arity of predicates in $\mathfrak P$,
a  relation  $\simeq^*_{i,j}$ which is similar to
$\simeq_{i,j}^0$, but for ancestor/descendant pairs instead of pairs of siblings.
We prefer to first define it in terms of $\adr{\_}$:

\begin{definition}
   $u \simeq^*_{i,j} w$ \;\iffi\; $u \earliereq w$ and $\adr{u}_i = \adr{w}_j$
   \end{definition}

\begin{observation}\label{obs:term-equality-decoding-works-ancestral}
For $u,w\in \theforest$, such that $u$ is a prefix of $w$, the following two conditions are equivalent:
    \begin{itemize}
\item  $u \simeq^*_{i,j} w$
\item  $\absadr{u}_i = \absadr{w}_j$  and the abstract term $\absadr{u}_i$ 
 occurs in atom $\absadr{v}$, for each address $v$ on the path from $u$ to $w$ - meaning $v$ is such that $u \earliereq v \earliereq w$.
\end{itemize}
\end{observation}

\begin{pr}
\hspace{-1mm} By induction, using Obs. \ref{lem:local-terms-belong-to-a-border-trigger}(i) and \cref{def:absadr}.\hspace{-4mm}\phantom{a}\qedhere
\end{pr}

\begin{observation}
    The relation $\simeq^*_{i,j}$ is MSO-expressible.
\end{observation}
\begin{prf}
    The formula expressing $u \simeq^*_{i,j} w$ is a disjunction over all abstract terms $n$:
    
    \smallskip
    \noindent
    $ \Exists{P} \absadr{u}_i \!=\! n \!=\! \absadr{w}_j \,\wedge\, \ancestralPath(P, u, w) \,\wedge\,  \Forall{v\!\in\!P} n \!\in\! \absadr{v}$ \qedhere
    
\end{prf}

Now recall \cref{lem:local-terms-belong-to-a-border-trigger}(ii):
\begin{definition}
   We define MSOL formula $ s \simeq_{i,j} u$  as a disjunction, 
   over all $i',j'$ 
      not greater than the maximal arity of predicates in $\mathfrak P$ of formulae:
$$\Exists{s',u'}\;\;  u' \simeq^*_{i',i} u \;\; \wedge \;\;  s' \simeq^*_{j',j} s \;\; 
\wedge \;\;
u' \simeq^0_{i',j'} s'
$$
\end{definition}

Combining all the above and \cref{lem:local-terms-belong-to-a-border-trigger}:
\begin{observation}\label{obs:simeq-works}
$ u \simeq_{i,j} s$ \;\iffi\; $\adr{u}_i = \adr{s}_j$.
\end{observation}
Finally, we are ready to express blocking of triggers:

\begin{lemma}\label{form:mso-blocking}
There exists an MSOL formula $\blocks_\rho(v_1, v_2, u)$ that holds \iffi
 $\pair{\rho, u}$ is a trigger blocked by the team $\bt{v_1, v_2}$.
\end{lemma}
\begin{prf}
Let $\rho$ be as in \cref{obs:trivial}. From that observation, it is enough to  show that one can express in MSOL the following conditions:
\begin{enumerate}[(i)]
    \item the predicate symbols of  $\adr{v_\iota}$ and of $\adr{u\rho_\iota}$ are equal;
    \item if $j$ is a frontier position in $\psi_\iota$ then $\adr{u\rho_\iota}_j= \adr{v_\iota}_j$;
    \item if $j$ is any  position in $\psi_\iota$ and $j'$ is any  position in $\psi_{\iota'}$ and $\adr{u\rho_\iota}_j=\adr{u\rho_{\iota'}}_{j'}$ then also $\adr{v_\iota}_j=\adr{v_{\iota'}}_{j'}$.
\end{enumerate}
(i) can be expressed as $\adr{w}$ and $\absadr{w}$ are isomorphic for any address, due to $\absadr{\_}$ correctness.
By \cref{obs:simeq-works} the statements (ii)--(iii) are trivially expressible using $\simeq_{i,j}$.
\end{prf}

This ends the difficult part of the proof; what is left should be a mere formality.
Consider some auxiliary formulae:
\begin{itemize}
    \item$\pair{\rho, u} \in \depriv \;=\; \Exists{u'\in \ppath} \rho_1(u,u') \lor \rho_2(u,u')$

    \item $\pblocks_\rho(v_1, v_2, u) = \pair{\rho, u} \in \depriv \;\land\; \blocks_\rho(v_1, v_2, u).$
    \item $\areTwins(u,v) \;=\;  
    \exists{w\iota} \bigvee_{\rho \in \rs} \rho_\iota(w, u) \land \rho_{3 - \iota}(w, v)$
    \item$u \in \deprivforest \; = \; u \in \ppath \;\lor\; \Exists{w} w \in \ppath \land \areTwins(u, w)$
    \item$v \in \preSet(u)= \Exists{v'} \areTwins(v, v') \land (v \earliereq u \lor v' \earliereq u) \lor \phantom{a}\quad\quad\,\quad\quad\quad \neg\Exists{w} w \earlier v.$
    \item $v \in \freeSet(u) = v \not\in \deprivforest \;\land\; u \not\earliereq v$
\end{itemize}

\noindent
Finally \quad
$\existsSensitivePath^0_{\omega} \quad=\quad \forall{u \in \ppath}\; B_1(u) \land B_2(u)$, \quad where

\medskip
\noindent
$B_1(u)= \neg \Exists{vv'\!\in\! \preSet(u)} \bigvee_{\rho\in\rs}\pblocks_{\rho}(v,v',u).$
\vspace{-2mm}%
\begin{align*}
&\phantom{a}\hspace{-4.5mm}B_2(u) =&&\; \Forall{W} \Big( \Forall{w \in W}\;\; u \in \preSet(w) \;\land\; \\&& &\Exists{v,s\in \freeSet(u) \cup \preSet(w)}  \bigvee_{\rho\in\rs}\pblocks_{\rho}(v,s,w) \Big)\\&& &\Rightarrow \neg\infinite(W).
\end{align*}

\subsection{Proof of \cref{lem:skonczone}}

\begin{customlem}{\ref{lem:skonczone}}
The sets $\harma(v)$ and $\harmb(v,s)$ are finite for all $v,s\in \freeSet(u)\cup \preSet(u)$.
\end{customlem}
\begin{prf} 
Follows from the fact that addresses involved in a blocking team must be created after the trigger they block in an $\omega$-sensitive-path-derivation. Thus, $\harma(v)$ and $\harmb(v,s)$ can only contain addresses $w$ such that $\pi_w$ occurs in $\demriv$ before the triggers creating $v$ and $s$.
\end{prf}